\newif\ifconf
\conffalse

\ifconf
\RequirePackage[T1]{fontenc}
\RequirePackage[ruled]{algorithm2e}
\PassOptionsToPackage{capitalize}{cleveref}
\fi

\ifconf
\documentclass[USenglish,thm-restate,cleveref]{lipics-v2021}
\else
\documentclass[11pt]{article}
\fi
\pdfoutput=1

\ifconf
\nolinenumbers
\else
\usepackage{microtype}
\usepackage[utf8]{inputenc}
\usepackage[T1]{fontenc}
\usepackage{amsmath}
\usepackage{amssymb}
\usepackage{lineno}
\usepackage[margin=1in]{geometry}
\usepackage{amthm}
\usepackage[ruled]{algorithm2e}
\usepackage{hyperref}
\usepackage[capitalize]{cleveref}
\fi

\ifconf\else
\usepackage{sectsty}
\subparagraphfont{\normalfont\emph}

\let\subparagraph\paragraph
\fi

\usepackage{complexity}
\usepackage{mathdots}  
\usepackage{csquotes}

\usepackage{xcolor}
\colorlet{lightgreen}{green!10!}

\usepackage[mode=buildnew,obeyclassoptions=true]{standalone}

\SetCommentSty{normaltext}

\crefname{line}{line}{lines} 
\SetKw{Accept}{accept}
\SetKw{Reject}{reject}
\SetKwData{Accepts}{accepts}
\SetKwData{NewState}{newState}
\SetKwData{StateCenter}{stateCenter}
\SetKwData{StateLeft}{stateLeft}
\SetKwData{StateRight}{stateRight}

\usepackage{tikz}
\usetikzlibrary{arrows,calc}

\ifconf
\makeatletter\let\c@author\relax\makeatother
\fi
\usepackage[
  backend=biber,bibencoding=utf8,sortcites=true,maxbibnames=99
]{biblatex}
\AtEveryBibitem{%
  \clearfield{isbn}
  \clearfield{issn}
  \clearfield{note}
  \clearfield{url}
}

\let\H\relax

\usepackage{ammacros}

\newcommand{\Ar}{_\mathrm{A}}
\newcommand{\CR}{_\mathrm{CR}}

\newcommand{\Gfunc}{G\colon\binalph^d \to \binalph^n}

\newcommand{\acc}{\mathrm{acc}}
\newcommand{\base}{_\mathrm{base}}
\newcommand{\il}{\mathbin{\parallel}}  
\newcommand{\intm}{_\mathrm{int}}

\newclass{\RPTIME}{RPTIME}

\title{Pseudorandom Generators for Sliding-Window Algorithms}

\newcommand{\myname}{Augusto Modanese}
\newcommand{\mynameabbr}{A.~Modanese}
\newcommand{\myaffil}{Aalto University\ifconf, Espoo, Finland\fi}
\newcommand{\myemail}{augusto.modanese@aalto.fi}
\newcommand{\myack}{I would like to thank Thomas Worsch for the helpful
  discussions and feedback.
  I also thank Ted Pyne for pointing out some constructions for pseudorandom
  generators that slightly improved the results.}

\ifconf
\author{\myname}{\myaffil}{\myemail}{}{}
\authorrunning{\mynameabbr}
\Copyright{\myname}
\acknowledgements{\myack}
\else
\author{\myname \\ \myaffil \\ \href{mailto:\myemail}{\myemail}}
\date{}
\fi

\begin{document}

\maketitle

\begin{abstract}
  A sliding-window algorithm of window size $t$ is an algorithm whose current
  operation depends solely on the last $t$ symbols read.
  We construct pseudorandom generators (PRGs) for low-space randomized
  sliding-window algorithms that have access to a binary randomness source.
  More specifically, we lift these algorithms to the non-uniform setting of
  branching programs and study them as a subclass thereof that we call
  \emph{sliding-window branching programs} (SWBPs), accordingly.
  For general SWBPs, given a base PRG $G\base$ with seed length $d\base$ that
  $\eps\base$-fools width-$w$, length-$t$ (general) branching programs, we give
  two PRG constructions for fooling any same-width SWBP of length $n$ and window
  size $t$ (where we assume $w \ge n$).
  The first uses an additional $d\base + O(\log(n/t) \log(1/\eps\base))$ random
  bits, whereas the second has a seed length of $O((d\base + \log\log(n/t) +
  \log(1/\eps\base)) \log(d\base + \log(1/\eps\base)))$.
  Both PRGs incur only a $(n/2t)^{O(1)}$ multiplicative loss in the error
  parameter.
  We also give a hitting set generator (HSG) that achieves a slightly better
  seed length for regimes where $\log w$ is very close to $\log(n/\eps)$ (e.g.,
  $\log w = (\log(n/\eps))^{1.01}$).

  As an application, we show how to decide the language of a sublinear-time
  probabilistic cellular automaton using small space.
  More specifically, these results target the model of \emph{PACAs}, which are
  probabilistic cellular automata that accept if and only if all cells are
  simultaneously accepting.
  For (sublinear) $T(n) = \Omega(\log n)^{1.01}$, we prove that every language
  accepted by a $T$-time one-sided error PACA (the PACA equivalent of $\RP$) can
  be decided using only $O(T)$ space.
  Meanwhile, forgoing the previous requirement on $T$, we show the same holds
  for $T$-time two-sided error PACA (the PACA equivalent of $\BPP$) if we use
  $\tilde{O}(T) + O(\log n)$ space instead (where the $\tilde{O}$ notation hides
  only $\polylog(T)$ factors).
\end{abstract}


\section{Introduction}

The processing of long streams of data using as little memory resources as
possible is a central computational paradigm of high relevance in the modern
age.
Through its presentation as \emph{streaming algorithms}, the topic has been the
subject of intense study within the realm of theoretical computer science.
When dealing with data that ages quickly, however, it appears a more accurate
representation of the process may be found in the subclass of
\emph{sliding-window algorithms}.
These are algorithms that maintain a window (hence the name) of only the last
few symbols in their stream and model processes where only the most recent data
is considered accurate or of relevance.
Natural examples of this strategy may be found in weather forecasting and social
media analysis.

A defining feature of sliding-window algorithms is
\emph{self-synchronization}:
If there are two machines executing the same algorithm with a window of size $t$
on the same stream and they receive faulty, inconsistent data at some point in
the process, then the machines will return to having identical states (i.e.,
synchronize) after having read at most $t$ identical symbols.
It turns out this behavior is not only potentially desirable in practice but,
from a theoretical point of view, it suggests itself as the defining
characteristic for these kind of algorithms.
In this work, we shall adopt this standpoint and consider \emph{low-space}
algorithms with this property that are \emph{augmented with access to a (binary)
randomness source}.
We will attempt to answer the following central question:
How much (if any) randomness is required by the class of algorithms with this
property in order to perform the tasks required of them?

\subparagraph{Randomized sliding-window algorithms.}
Intuitively, we should obtain a randomized version of sliding-window algorithms
by having the control unit operate according to a randomized process.
Of course, this immediately raises the question of how one should adapt the
sliding-window requirement to this new model.
As the defining characteristic (in the deterministic case) is that the machine's
current operation \emph{depends only on the last $t$ symbols read}, it is
plausible to require the algorithm's behavior to be dependent not only on the
last $t$ symbols \emph{but also on the random choices made while reading said
bits}.
Since a true randomized algorithm (in our setting) is expected to use at least 
one bit of randomness for every new symbol it reads, it appears as if the
algorithm's window would in its greater part (or at least as much so) contain
bits originating from the randomness source rather than the input stream.
Actually, we will altogether abandon the latter when generalizing the
sliding-window property to randomized algorithms; that is, we will phrase the
property exclusively in terms of the randomness source and, using
non-uniformity, simply forgo any mention of the input stream.
This approach has a couple of advantages:
\begin{enumerate}
  \item It not only simplifies the presentation (since then the algorithm passes
  its window over a single input source instead of two) but also renders the
  model more amenable to the usual complexity-theoretical methods for analyzing
  low-space algorithms.
  \item The resulting class of machines is actually \emph{stronger} than the
  model where the sliding-window property applies to both the data stream and
  the randomness source.
  Hence, our results not only hold in the latter case but also in the more
  general one.
\end{enumerate}

As mentioned in passing above, we will rely on classical methods from complexity
theory to analyze low-space algorithms in the context of derandomization, namely
the non-uniform model of \emph{branching programs}.

\subsection{Branching Programs}

Branching programs can be seen as a non-uniform variant of deterministic finite
automata and have found broad application in the derandomization of low-space
algorithms.

\begin{definition}
  An \emph{(ordered, read-once) branching program} $P$ of length $n$ and width
  $w$ is a set of states $Q$, $\abs{Q} = w$, of which a subset $Q_\acc \subseteq
  Q$ are \emph{accepting states}, along with $n$ transition functions
  $P_1,\dots,P_n\colon Q \times \binalph \to Q$.
  The program $P$ starts its operation in an \emph{initial state} $q_0 \in Q$
  and then processes its input $x = x_1 \cdots x_n \in \binalph^n$ from left to
  right while updating its state in step $i$ according to $P_i$.
  That is, if $P$ is in state $q_{i-1}$ after having read $i-1$ bits of its
  input, then it next reads $x_i$ and changes its current state to
  $P_i(q_{i-1},x_i)$.
  We say $P$ \emph{accepts} $x$ if $q_n \in Q_\acc$, where $q_n$ is state of $P$
  after processing the final input symbol $x_n$.
\end{definition}

Although $P$ reuses the same set of states throughout its processing of $x$, it
is also natural to view $P$ as a directed acyclic graph with $n+1$
\emph{layers}, where the $i$-th layer $Q_i$ contains a node for every state of
$Q$ that is reachable in (exactly) $i$ steps of $P$.
(We shall refer to $Q_i$ as a set of nodes and as a subset of $Q$
interchangeably.)
A node $v \in Q_{i-1}$ is then connected to the nodes $u_0 = P_i(v,0)$ and $u_1
= P_i(v,1)$ of $Q_i$ and we label the respective edges with $0$ and $1$.
(If $u_0 = u_1$, then we have a double edge from $v$ to $u_0$.)

We write $P(x)$ for the indicator function $\binalph^n \to \binalph$ of $P$
accepting $x$ (i.e., $P(x) = 1$ if $P$ accepts $x$, or $P(x) = 0$ otherwise).
For $y \in \binalph^\ast$ with $\abs{y} \le n$, we write $P_0(y)$ to indicate
the \emph{state} of $P$ after having read $y$ when starting its operation in its
initial state $q_0$.
Letting $\lambda$ denote the empty word, this can be defined recursively by
setting $P_0(\lambda) = q_0$ and $P_0(y'z) = P_{\abs{y'}+1}(P_0(y'),z)$ for $y =
y'z$ and $z \in \binalph$.
We extend the domain of $P_i$ from $Q \times \binalph$ to $Q \times
\binalph^{\le n-i}$ in the natural way as follows: $P_i(q,\lambda) = q$ and
$P_i(q,yy') = P_{i+1}(P_i(q,y),y')$ for $yy' \in \binalph^{\le n-i}$ where $y
\in \binalph$.

\emph{Unanimity programs} \cite{bogdanov22_hitting_ccc} are a generalization of
branching programs that accept if and only if every state during its computation
is accepting (instead of only its final one).
Formally, this means we mark a subset $Q_\acc^i \subseteq Q_i$ of every $i$-th
layer as accepting and say that a unanimity program $U$ \emph{accepts} $x \in
\binalph^n$ if and only if $U_0(y) \in Q_\acc^{\abs{y}}$ for every prefix $y
\neq \lambda$ of $x$.
(Note this is indeed a generalization since setting $Q_\acc^i = Q$ for every $i
< n$ yields a standard branching program.)
It is easy to see that a width-$w$ unanimity program can be simulated by a
width-$(w+1)$ standard branching program (e.g., by adding a \enquote{fail}
state to indicate the unanimity program did not accept at some point).

\subparagraph{Pseudorandom generators.}
The key concept connecting branching programs and the topic of derandomization
is that of \emph{pseudorandom generators}.
These are general-purpose functions that take a (conceptually speaking) small
subset of inputs and \enquote{scatter} them across their range in such a way
that \enquote{appears random} to the class of procedures they intend to fool.
In the definition below, $U_n$ denotes a random variable distributed uniformly
over $\binalph^n$.

\begin{definition}%
  \label{def_prg}%
  Let $n \in \N_+$ and $\eps > 0$, and let $\mathcal{F}$ be a class of functions
  $f\colon \binalph^n \to \binalph$.
  We say a function $G\colon \binalph^d \to \binalph^n$ is a \emph{pseudorandom
    generator} (PRG) that \emph{$\eps$-fools} $\mathcal{F}$ if the following
  holds for every $f \in \mathcal{F}$:
  \[
    \abs*{\Pr[f(G(U_d)) = 1] - \Pr[f(U_n) = 1]} \le \eps.
  \]
  In addition, we say $G$ is \emph{explicit} if there is a (uniform)
  linear-space algorithm which, on input $n$ (encoded in binary) and $s \in
  \binalph^d$, outputs $G(s)$.
\end{definition}

The connection between the above and low-space algorithms is as follows:
Suppose we have an $s$-space algorithm $A$ that takes as input a string $x$
along with a stream $r \in \binalph^n$ of random coin tosses.
If we take a PRG $G$ that $\eps$-fools branching programs of width $2^s$ and
length $n$, then $A(x,G(U_d))$ must be $\eps$-close to $A(x,U_n)$.
To see why, notice that we can convert $A(x,\cdot)$ into a branching program
$P_x$ by taking as state set for $P_x$ the configuration space of $A$ and
non-uniformly hardcoding the bits of $x$ read by $A$ (into the transition
function of $P_x$).
The resulting program $P_x$ then takes as input the random stream of $A$ (which,
despite being perhaps counter-intuitive, yields the desired reduction).

One of the ultimate goals behind the study of PRGs for branching programs is the
resolution of the question of whether $\L$ equals $\BPL$.
Non-constructively, it can be shown that there exists a PRG with \emph{optimal}
seed length $O(\log(w/\eps))$ that $\eps$-fools any width-$w$ branching program
of arbitrary length.
If we could show there is a PRG with these parameters that is \emph{explicit},
then this would imply $\L = \BPL$.
The currently best known derandomization of $\L$ (ignoring sublogarithmic
factors) uses $O(\log n)^{3/2}$ space \cite{saks99_BPHSpace_jcss}.
Our results are interesting in this context as we obtain PRGs having optimal or
near-optimal (in the sense of being optimal up to sublogarithmic factors) seed
length in the special case of sliding-window algorithms.

\subparagraph{Sliding-window branching programs.}
In this work, our goal is to construct PRGs against the restricted class of
unanimity programs satisfying the sliding-window property.
As previously discussed, the defining feature is self-synchronization.

\begin{definition}
  A \emph{sliding window branching program} (SWBP) is a unanimity program $S$
  that satisfies the following property:
  There is a number $t \in \N_+$, called the \emph{window size} of $S$, such
  that, for every $i \le n-t$, every $y \in \binalph^t$, and every pair of
  states $q,q' \in Q_i$, we have that $S_i(q,y)  = S_i(q',y)$.
\end{definition}

In other words, the current state of $S$ depends exclusively on the last $t$
bits read.
From this perspective, it is not hard to see that a generalization to unanimity
programs is actually necessary for the sliding-window property to be
interesting.
Indeed, if an SWBP $S$ is a standard branching program (i.e., not just a
unanimity program), then its decision only depends on the last $t$ bits of its
input.

The relation between SWBPs and sliding-window algorithms is as the one described
for general branching programs and low-space algorithms following
\cref{def_prg}.
We refer the reader to \cref{sec_intro_related_work} further below for the
preexisting literature on sliding-window algorithms.

For the sake of motivation, let us give an example of a very natural type of
randomized algorithm having this property.
(We present it in a rather abstract manner since it encompasses a whole family
of algorithms based on the this scheme.)
Suppose we are given some input $x$ and can perform a low-space randomized
procedure $T$ that tests some property of $x$ (perhaps conditioned on some
auxiliary input).
Then there is a natural sliding-window algorithm $A = A_T$ based on $T$ that
simply executes $T$ on $x$ a given number of times $m = m(\abs{x})$ (using
independent coin tosses for each execution) and accepts if and only if every
every execution of $T$ succeeds.
The sliding-window property comes from the fact that each execution of $T$ is
independent from another and $A$ is only interested in whether $T$ returns a
positive answer (i.e., $A$ does not adapt its next test to previous test
results).

\subsection{Our Results}

Our central result is the construction of optimal or near-optimal PRGs for the
SWBP model.
As an application, we derandomize the languages accepted by a certain model of
sublinear-time probabilistic cellular automata.
In addition, we give a structural characterization of SWBPs, which we address
next.

\subsubsection{Structural Characterization of SWBPs}
\label{sec_intro_swbp_structure}

Let $n \in \N_0$.
Recall the \emph{de Bruijn graph} of dimension $n$ is the directed graph $B_n =
(V_n,E_n)$ with vertex set $V_n = \binalph^n$ and edge set
\[
  E_n = \{ (xw,wy) \mid \text{$w \in \binalph^{n-1}$ and $x,y \in \binalph$} \}.
\]
Similarly, the \emph{prefix tree} of dimension $n$ is the graph $T_n =
(V_n',E_n')$ with $V_n' = \binalph^{\le n}$ and
\[
  E_n' = \{ (w,wx) \mid \text{$w \in \binalph^{\le n-1}$
    and $x \in \binalph$} \}.
\]

As our first contribution, we observe that, for every fixed window size $t$ and
input length $n$, there is a \enquote{prototypical} SWBP $\Pi$ whose topology is
such that its first $t$ layers are isomorphic to the prefix tree $T_t$ and its
subsequent layers are connected according to (the labeled version of) the de
Bruijn graph $B_t$.
Thus, every SWBP $S$ of window size $t$ can be obtained from this prototypical
SWBP $\Pi$ by merging nodes in the same layer.
(Note this process may also incur merging nodes in subsequent layers as well so
as to ensure that, for every $i$, $S_i$ is indeed a function.)

\begin{theorem}[restate=restatethmCharSWBP,name=]%
  \label{thm_char_swbp}%
  A unanimity program $S$ is a SWBP of window size $t$ if and only if there are
  functions $\alpha_0,\dots,\alpha_n\colon \binalph^{k_i} \to Q$ with
  \[
    k_i = \begin{cases}
      i, & i<t; \\
      t, & i \ge t
    \end{cases}
  \]
  and such that the following are true for every $x,y \in \binalph$ and
  $w \in \binalph^{k_i-1}$:
  \begin{enumerate}
    \item For every $i < t$, $S_i(\alpha_i(w),y) = \alpha_{i+1}(wy)$.
    \item For every $t \le i < n$, $S_i(\alpha_i(xw),y) = \alpha_{i+1}(wy)$.
  \end{enumerate}
\end{theorem}

From this characterization it immediately follows that, as long as $S$ does not
have redundant states, we have $w \le 2^t$.

\subsubsection{PRGs for SWBPs}
\label{sec_intro_prgs_swbp}

We give two constructions for a PRG that fools SWBPs, where the one or the other
might be preferable depending on the regime of parameters at hand.

\begin{theorem}[restate=restatethmPRGGenSWBP,name=]%
  \label{thm_prg_gen_swbp}%
  Let $n,w,t \in \N_+$ with $n \le w$ and $\eps > 0$ be given, and let
  $G\base\colon \binalph^{d\base} \to \binalph^t$ be a PRG that
  $\eps\base$-fools width-$w$, length-$t$ unanimity programs.
  Then the following holds:
  \begin{enumerate}
    \item There is an explicit PRG $\Gfunc$ with seed length
    \[
      d = 2d\base + O(\log(n/t)\log(1/\eps\base))
    \]
    that $\eps_G$-fools any width-$w$, length-$n$ SWBP of window size $t$, where
    $\eps_G = \eps\base \cdot (n/t)^{O(1)}$.
    \item There is an explicit PRG $\Gfunc$ with seed length
    \[
      d' = O(d\base + \log(n/t) + (\log(1/\eps\base))^{3/2})
    \]
    that $\eps_G$-fools any width-$w$, length-$n$ SWBP of window size $t$, where
    $\eps_G = O(\eps\base \cdot n / t)$.
    \item There is an explicit PRG $\Gfunc$ with seed length
    \[
      d' = O((d\base + \log\log(n/t) + \log(1/\eps\base))
        \log(d\base + \log(1/\eps\base)))
    \]
    that $\eps_G$-fools any width-$w$, length-$n$ SWBP of window size $t$, where
    $\eps_G = O(\eps\base \cdot n / t)$.
  \end{enumerate}
\end{theorem}

To get a feeling for how this is an improvement compared to a PRG that fools
general branching programs, consider first the case in which $G\base$ is the
classical generator of \citeauthor{nisan92_pseudorandom_comb}
\cite{nisan92_pseudorandom_comb}. 
Assuming $n \le w$, this generator $\eps\base$-fools any branching program of
width-$w$ and length-$n$ using a seed length of
\[
  d\base = O(\log(w/\eps\base) \log n).
\]
Assume also $t = o(n)$, which is generally the most interesting case.
Then the second and third items of \cref{thm_prg_gen_swbp} give us PRGs with
seed lengths
\[
  d_2 = O(\log(w/\eps)\log t + (\log(n/\eps))^{3/2})
\]
and
\[
  d_3 = O(\log(w/\eps) \log\log(w/\eps) \log t)
\]
for any desired $\eps > 0$, respectively.
(To keep the discussion brief, we set the first item of \cref{thm_prg_gen_swbp}
aside for now, but note it gives a better dependence for smaller values of
$\eps$ compared to $d_2$.)
For the regime where $\log(w/\eps) = \Omega(\log(n/\eps))^{3/2}$, this means the
$\log n$ factor in $d\base$ becomes a $\log t$ factor in $d_2$.
This is a considerable improvement for most applications since the window size
$t$ is likely to be much smaller than the input length $n$.
Meanwhile, $d_3$ is independent of $n$, suggesting it is more appropriate when
fooling programs whose width is very close to $n$, corresponding to algorithms
of low space complexity (i.e., near $\log n$).

The best known generator for fooling general branching programs is the one by
\citeauthor{armoni98_derandomization_random}
\cite{armoni98_derandomization_random}, who gives an explicit generator
$G\Ar\colon \binalph^{d\Ar} \to \binalph^n$ that (with the improvements by
\citeauthor{kane08_revisiting_arxiv} \cite{kane08_revisiting_arxiv})
$\eps\Ar$-fools any width-$w$, length-$n$ branching program for any choice of
$\eps\Ar > 0$ and $n,w \in \N_+$ (and, in particular, also $\eps\Ar$-fools any
width-$(w-1)$, length-$n$ unanimity program).
Again assuming $n \le w$, $G\Ar$ has seed length
\[
  d\Ar = O\left(
    \frac{\log(w/\eps\Ar) \log n}{\max\{1,\log\log w - \log\log(n/\eps\Ar)\}}
  \right).
\]
Given $\eps > 0$, plugging in $G\base = G\Ar$ in \cref{thm_prg_gen_swbp} with
$\eps\base = \eps\Ar = \eps \cdot (2t/n)^{\Omega(1)}$, we get:

\begin{corollary}%
  \label{cor_thm_prg_gen_swbp}%
  Let $n,w,t \in \N_+$ with $n \le w$ and $\eps > 0$ be given.
  Furthermore, let
  \[
    d' = O\left(
      \frac{\log(w/\eps) \log t}{\max\{1,\log\log w - \log\log(n/\eps)\}}
    \right).
  \]
  Then the following hold:
  \begin{enumerate}
    \item There is an explicit PRG $\Gfunc$ with seed length
    \[
      d = O(d' + \log(n/t)\log(n/t\eps))
    \]
    that $\eps$-fools any width-$w$, length-$n$ SWBP of window size $t$.
    \item There is an explicit PRG $\Gfunc$ with seed length
    \[
      d = O(d' + (\log(n/t\eps))^{3/2})
    \]
    that $\eps$-fools any width-$w$, length-$n$ SWBP of window size $t$.
    \item There is an explicit PRG $\Gfunc$ with seed length
    \[
      d = O((d' + \log(n/t\eps)) \log(d' + \log(n/t\eps)))
    \]
    that $\eps$-fools any width-$w$, length-$n$ SWBP of window size $t$.
  \end{enumerate}
\end{corollary}

If we take the second item and consider the regime where (again) $\log w =
\Omega(\log(n/\eps))^{3/2}$ and the window size $t$ is in $\polylog(w)$, then
the resulting seed length is $O(\log(w/\eps))$, which is essentially optimal.
Meanwhile, the PRG from the first item is better suited for similar regimes but
smaller values of $\eps$ (e.g., for $\eps = 2^{-\Omega(\log n)^2}$).
In turn, the PRG from the third one is near-optimal for those regimes where
$\log w$ is in-between $\Omega(\log(n/\eps))$ and $O(\log(n/\eps))^{3/2}$, which
allows us to almost fully close the gap left for algorithms in the lower end of
the space complexity spectrum.

\subparagraph{Hitting set generators for SWBPs.}
With a slight adaptation of our second of the two constructions above, we are
able to improve the seed length if we relax it from a PRG to just a
\emph{hitting set generator} (HSG).
In contrast to a PRG, an HSG does not have to appear random all across its
range to the class $\mathcal{F}$ of functions we wish to fool; it only has to
guarantee that the functions in $\mathcal{F}$ cannot avoid its range completely
(again, up to some margin of error $\eps$).

\begin{definition}
  Let $n \in \N_+$ and $\eps > 0$, and let $\mathcal{F}$ be a class of functions
  $f\colon \binalph^n \to \binalph$.
  We say a function $G\colon \binalph^d \to \binalph^n$ is an \emph{$\eps$-HSG}
  for $\mathcal{F}$ if the following holds for every $f \in \mathcal{F}$:
  \[
    \Pr[f(U_n) = 1] \ge \eps
    \implies \exists x \in \binalph^d: f(G(x)) = 1.
  \]
  Again, we say $G$ is \emph{explicit} if there is a (uniform) linear-space
  algorithm which, on input $n$ (encoded in binary) and $s \in \binalph^d$,
  outputs $G(s)$.
\end{definition}

\begin{theorem}[restate=restatethmHSGSWBP,name=]%
  \label{thm_hsg_swbp}%
  Let $n,w,t \in \N_+$ with $n \le w$ and $\eps > 0$ be given, and let
  $G\base\colon \binalph^{d\base} \to \binalph^t$ be a PRG that
  $\eps\base$-fools width-$w$, length-$t$ unanimity programs.
  Then there is an explicit $\eps_G$-HSG $\Gfunc$ with seed length
  \[
    d = O(d\base + \log(n/t) + \log(1/\eps\base))
  \]
  for width-$w$, length-$n$ SWBPs of window size $t$, where $\eps_G = O(\eps\base
  \cdot (n/t))$.
\end{theorem}

Plugging in Armoni's PRG as before, we have:

\begin{corollary}%
  \label{cor_thm_hsg_swbp}%
  Let $n,w,t \in \N_+$ with $n \le w$ and $\eps > 0$ be given.
  Furthermore, let
  \[
    d' = O\left(
      \frac{\log(w/\eps) \log t}{\max\{1,\log\log w - \log\log(n/\eps)\}}
    \right).
  \]
  Then there is an explicit $\eps$-HSG $\Gfunc$ with seed length
  \[
    d = O(d' + \log(n/t) + \log(1/\eps))
  \]
  for width-$w$, length-$n$ SWBPs of window size $t$.
\end{corollary}

Since \cref{cor_thm_prg_gen_swbp} already gives us a PRG with comparable seed
length (which is more desirable than obtaining an HSG), we focus here on the
regime where $\log w = o(\log(n/\eps))^{3/2}$.
Similar to the discussion following \cref{cor_thm_prg_gen_swbp}, if $t =
\polylog(w)$, then this gives us an HSG with essentially optimal seed length
$O(\log(w/\eps))$.
The difference is that this now holds even in regimes where $\log w =
\Omega(\log(n/\eps))^{1+\delta}$ for some $\delta > 0$ (and not only $\log w =
\Omega(\log(n/\eps))^{3/2}$ as before).

\subsubsection{Application to Probabilistic Cellular Automata}
\label{sec_intro_prgs_paca}

As an application, we obtain space-efficient algorithms for deciding the
languages accepted by sublinear-time probabilistic cellular automata.
More specifically, our results target the model of PACAs (probabilistic ACAs)
\cite{modanese23_sublinear-time_stacs}, which are cellular automata with two
local transition functions $\delta_0$ and $\delta_1$ and where, at every step,
each cell tosses a fair coin $c \in \binalph$ and then updates its state
according to $\delta_c$.
The acceptance condition is that of ACA, which is the most usual one
\cite{ibarra85_fast_tcs, sommerhalder83_parallel_ai, kim90_characterization_pd,
modanese21_sublinear-time_ijfcs} that allows for non-trivial sublinear-time
computations:
A computation is accepting if and only if a configuration is reached in which
every cell is accepting.
Our results target both \emph{one-sided} and \emph{two-sided error} PACAs, which
in a sense are the PACA analogues of the classical complexity classes $\RP$ and
$\BPP$, respectively.
(We refer the reader to \cref{sec_def_paca} for the definitions.)

These results are interesting because similar positive results in the sister
setting of efficiently deciding the language of a PACA where efficiency is
measured in terms of \emph{time} complexity would have surprising consequences
for the derandomization of Turing machines (e.g., $\P = \RP$)
\cite{modanese23_sublinear-time_stacs}.

\subparagraph{One-sided error PACAs.}
Previous to this work, it was known that every language accepted by a
deterministic ACA with time complexity $T$ can be decided using $O(T)$ space
\cite{modanese21_sublinear-time_ijfcs}.
For general one-sided error PACAs, we obtain the following result:
\begin{theorem}[restate=restatethmPRGGenPACA, name=]%
  \label{thm_prg_gen_paca}
  Let $\eps > 0$ be arbitrary, and let $T\colon \N_+ \to \N_+$ be a
  (constructible) function.
  Then, for any one-sided $\eps$-error PACA $C$ that recognizes its language
  $L(C)$ in time at most $T = T(n)$, there is a deterministic algorithm for
  $L(C)$ with space complexity
  \[
    O\left(
      \frac{(T + \log(1/\eps)) \log T}{\max\{1, \log T - \log\log(n/\eps)\}}
      + \log(1/\eps)
      + \log n
    \right).
  \]
\end{theorem}
Hence, for (say) $T(n) \ge (\log n)^{1.01}$ and constant $\eps$, it follows that
using $O(T)$ space we can also decide any language accepted by a $T$-time
(one-sided error) \emph{probabilistic} ACA.
That we can specify $\eps$ here as arbitrary (in particular allowing it to
depend on $n$) is a nice plus because we only know direct constructions for
reducing the error of a PACA by a \emph{constant} factor
\cite{modanese23_sublinear-time_stacs}.

\subparagraph{Two-sided error PACAs.}
For two-sided error PACAs, we fall only very short of achieving the same space
complexity as above.
\begin{theorem}[restate=restatethmPRGGenPACATwosided,name=]%
  \label{thm_prg_gen_paca_twosided}%
  Let $T\colon \N_+ \to \N_+$ be a function and $\eps > 0$ be arbitrary.
  For any two-sided $(1/2 - \eps)$-error PACA $C$ that recognizes its language
  $L(C)$ in time at most $T = T(n)$, there is a deterministic algorithm for
  $L(C)$ with space complexity $O((d' + \log n) \log(d' + \log n))$, where $d' =
  O((T + \log(1/\eps)) \log T)$.
\end{theorem}
Hence, this gives us that any language that can be accepted by a $T$-time
two-sided error PACA (i.e., with constant error $\eps$) can be decided
deterministically using $O(T(\log T)^2)$ space.
In fact, the same space complexity bound holds even if we only have $\eps =
2^{-O(T)}$.
For comparison, if we use Armoni's PRG directly, the resulting space complexity
(for constant $\eps$) is $O(T \log n)$; hence, ours is an improvement for PACA
that run in time $2^{O(\sqrt{\log n})}$.

\subsection{Technical Overview}
\label{sec_intro_tech_overview}

\subparagraph{Generators for general SWBPs.}
Our construction relies on two main ideas, the first of which is especially
suited for exploiting the sliding-window property.
With this technique, which we call \emph{interleaving}, we are able to
effectively \enquote{shatter} SWBPs into a collection of short programs.
To illustrate the idea, take some SWBP $S$ of window size $t$ and consider the
behavior on $S$ on an input
\[
  X \il Y = X_1 Y_1 X_2 Y_2 \cdots X_n Y_n
\]
where the $X_i$ and $Y_i$ all have length $t$ and $X = X_1 \cdots X_n$ and $Y =
Y_1 \cdots Y_n$ are chosen independently from one another (but possibly
follow the same distribution).
The point is that, since $S$ must read the $t$ bits of $Y_i$ between reading
$X_i$ and $X_{i+1}$, when it starts processing $X_{i+1}$, it has essentially
\enquote{forgotten} all information about $X_i$.
In addition, as $X$ and $Y$ are chosen independently from one another, $Y_i$
contains no information about $X_{i+1}$ whatsoever; the only relevance $Y_i$ has
regarding $X_{i+1}$ is in determining which state the processing of $X_{i+1}$
starts in.
Hence, when the output of our generator is of the form $X \il Y$, we can simply
(say) set $Y$ to some fixed string $y$ and then analyze $S$ as a collection of
$n$ many length-$t$ unanimity programs, each of which receives an $X_i$ as
input.

We have thus reduced our original task to that of fooling a collection of $n$
programs \emph{simultaneously}; that is, given unanimity programs
$P_1,\dots,P_n$ of the same width and length, we wish to generate pseudorandom
inputs $X_1,\dots,X_n$ so that
\[
  \abs*{\Pr[\forall i \in [n]: P_i(X_i) = 1] - \prod_{i=1}^n \Pr[P_i(U) = 1]}
\]
is small, where $U$ is the uniform distribution on all possible inputs to $P_i$.
As it turns out, there are \emph{two} approaches we may now apply, either one
being more advantageous in different parameter regimes.

\subparagraph{Simultaneous fooling as in the INW generator.}

The first construction we give resorts to one of the key ideas behind the
generator by \citeauthor{impagliazzo94_pseudorandomness_stoc}
\cite{impagliazzo94_pseudorandomness_stoc}:
Supposing we have a generator $G$ that simultaneously fools $n/2$ many programs
using a random seed $s$ of length $d$, we use an adequate \emph{extractor}
(\cref{thm_gw_extractor}) to generate a fresh seed $s'$ using \enquote{just a
few more} random bits and then output the concatenation $G(s) G(s')$. 
(Refer to \cref{sec_preliminaries} for the definitions.)
In our case, this strategy actually works \emph{as best as it possibly can}.
This is because only very little entropy is given away when the first $n/2$
programs have processed $G(s)$; namely, all that can be learned about $s$ is
that the $P_1,\dots,P_{n/2}$ are all accepting, which amounts to only a
\emph{constant} amount of information.

By induction, this allows us to simultaneously fool an arbitrary number $n \ge
2$ of unanimity programs.
For the base case $n=1$, we plug in any generator $G\base$ for unanimity
programs of our liking, the point being that $G\base$ only has to fool programs
that are as long as the \emph{window size} $t$ of our original SWBP $S$ (rather
than programs that are as long as $S$ itself).

\subparagraph{Simultaneous fooling using PRGs for combinatorial rectangles.}
The second item of \cref{thm_prg_gen_swbp} requires a different approach.
The general idea is reminiscent of a work by \citeauthor{hoza20_simple_siamjc}
\cite{hoza20_simple_siamjc}, though it is likely there are similar constructions
to be found in previous literature.
Instead of trying to \enquote{recycle} our seed again and again as we did
before, we will now pick $m$ many independent seeds $s_1,\dots,s_n$ and output
the concatenation $G\base(s_1) \cdots G\base(s_n)$.
This is guaranteed to fool the $n$ programs simply by construction (where the
price to pay is a very moderate factor of $n$ in the error parameter).
The next, key step is then to reduce the seed length by \emph{generating the
seeds themselves} with another construction that fools a different kind of
combinatorial object, namely so-called \emph{combinatorial rectangles}.

\begin{definition}
  Let $m,n \in \N_+$.
  An \emph{$(n,m)$-combinatorial rectangle} is a function $f\colon
  (\binalph^m)^n \to \binalph$ that results from composing the product of $n$
  many Boolean functions $f_i\colon \binalph^m \to \binalph$; that is,
  \[
    f(x_1,\dots,x_n) = \prod_{i=1}^n f_i(x_i).
  \]
\end{definition}

This is a good choice for two reasons:
Firstly, these objects are a tight fit to our notion of simultaneously fooling.
If we associate each $f_i$ with a program $P_i$, then $f(x_1,\dots,x_n) = 1$
corresponds exactly to the collection of programs accepting all respective
inputs $x_i$ simultaneously; moreover, just like our $P_i$, the output of each
$f_i$ is fully independent from one another.
Secondly, the literature on pseudorandom constructions targeting these objects
is relatively extense and, unlike general branching programs, asymptotically
optimal or near-optimal results are known (see \cref{thm_gy_prg,thm_llsz_hsg} in
\cref{sec_general_prg,sec_hsg}, respectively).

Finally, to obtain the HSG of \cref{thm_hsg_swbp}, we replace the PRG for
combinatorial rectangles above with the (asymptotically) optimal construction of
an HSG due to \citeauthor{linial97_efficient_comb}
\cite{linial97_efficient_comb}.
The analysis is similar, though the readaptation still requires a bit of work.
We refer to \cref{sec_hsg} for the details.

\subsection{Related Work}
\label{sec_intro_related_work}

\subparagraph{Branching programs.}
The standard line of attack in complexity theory when derandomizing low-space
algorithms is to lift these to the more general model of (non-uniform) branching
programs, which are more amenable to a combinatorial analysis.
This approach can be traced at least 30 years back to the seminal work of
\citeauthor{nisan92_pseudorandom_comb} \cite{nisan92_pseudorandom_comb}.
Since then, there has been progress in derandomizing branching programs in
diverse settings including, for instance, branching programs in which the
transition function at each layer is a permutation
\cite{hoza21_pseudorandom_itcs} or that may read their input in some fixed but
unknown order \cite{forbes18_pseudorandom_focs}.
Unanimity programs were recently proposed by
\citeauthor{bogdanov22_hitting_ccc} \cite{bogdanov22_hitting_ccc}.
To the best of our knowledge, ours is the first work to study branching programs
with the sliding-window property or any similar one, for that matter.

\subparagraph{Sliding-window algorithms.}
The sliding-window paradigm is a natural form of stream processing that has
been considered in the context of
database management systems \cite{babcock02_models_pods}, network monitoring
\cite{cormode13_continuous_sigmodr}, and reinforcement learning
\cite{garivier11_upper-confidence_alt}.
Starting with the work of \citeauthor{datar02_maintaining_siamjc}
\cite{datar02_maintaining_siamjc}, the sliding-window model has also been
extensively studied in the context of maintaining statistics over data streams.
(See, e.g., \cite{braverman16_sliding_book} for a related survey.)

Sliding-window algorithms have also been studied by
\citeauthor*{ganardi21_derandomization_tcs} \cite{ganardi18_randomized_icalp,
ganardi18_sliding_lata, ganardi18_sliding_mfcs, ganardi21_derandomization_tcs}
in the setting of language recognition (among others).
We point out a couple fundamental differences between the model we consider and
theirs:
\begin{itemize}
  \item Their results also apply to the non-uniform case---but parameterized on
  the \emph{window size}.
  In particular, their model allows radically different behaviors on the same
  stream for different window sizes.
  In our case, non-uniformity is parameterized on the \emph{input size} (i.e.,
  the length of the data stream).
  \item The underlying probabilistic model in the work of
  \citeauthor*{ganardi21_derandomization_tcs} is the \emph{probabilistic
  automata} model of \citeauthor{rabin63_probabilistic_ic}
  \cite{rabin63_probabilistic_ic}, which allows state transitions according to
  arbitrary distributions.
  In contrast, our model draws randomness from a binary source and---most
  importantly---the sliding-window property \emph{applies to the random input}
  (whereas in the model studied by \citeauthor*{ganardi21_derandomization_tcs}
  it only applies to the data stream).
\end{itemize}
In summary, \citeauthor*{ganardi21_derandomization_tcs} focus on a model that
verifies (or computes some quantity for) every window of fixed size on its
stream; we focus on a model that verifies \emph{the stream as a whole}.

We also mention a recent paper by \citeauthor{pacut21_locality_arxiv}
\cite{pacut21_locality_arxiv} that points out a connection between distributed
and sliding-window algorithms.
One may consider our application to probabilistic cellular automata to be a
direct consequence of this connection.

\subparagraph{Probabilistic cellular automata.}
Our contribution to probabilistic cellular automata adds another link to a
recent chain of results
\cite{modanese21_sublinear-time_ijfcs,modanese21_lower_csr,
modanese23_sublinear-time_stacs} targeted at the study of sublinear-time
cellular automata.
As mentioned in \cite{modanese21_sublinear-time_ijfcs}, the topic has been
seemingly neglected by the cellular automata community at large and, as far as
we are aware of, the body of theory on the subject predating this recent series
of papers resumes itself to \cite{ibarra85_fast_tcs, kim90_characterization_pd,
sommerhalder83_parallel_ai}.
A probabilistic model similar to the probabilistic cellular automata we consider
was previously proposed by \citeauthor{arrighi13_stochastic_fi}
\cite{arrighi13_stochastic_fi}, but \cite{modanese23_sublinear-time_stacs} is
the first work that addresses the sublinear-time case.

Finally, note we follow \cite{modanese23_sublinear-time_stacs} in using the term
\enquote{probabilistic} (due to their similarity to probabilistic Turing
machines) to refer to these automata and treat them separately from the more
general \emph{stochastic cellular automata} in which the local transition
function may follow an arbitrary distribution (in the same spirit as the
aforementioned work by \citeauthor{rabin63_probabilistic_ic}
\cite{rabin63_probabilistic_ic}).
Unfortunately, there is no consensus on the distinction between the two terms
in the literature, and the two have been used interchangeably.
For a survey on stochastic cellular automata, see \cite{mairesse14_around_tcs}.

\subsection{Organization}

The rest of the paper is organized as follows:
In \cref{sec_preliminaries}, we recall the basic definitions and results that we
need.
The subsequent sections each cover one set of results:
\cref{sec_structure} addresses the structural result on SWBPs.
\cref{sec_general_prg,sec_hsg} contain the proofs of
\cref{thm_prg_gen_swbp,thm_hsg_swbp}, respectively.
Finally, \cref{sec_paca} covers the applications to probabilistic cellular
automata.


\section{Preliminaries}
\label{sec_preliminaries}

It is assumed the reader is familiar with basic notions of computational
complexity theory and pseudorandomness (see, e.g., the standard references
\cite{goldreich08_computational_book, arora09_computational_book,
  vadhan12_pseudorandomness_book}).

All logarithms are to base $2$.
The set of integers is denoted by $\Z$, that of non-negative integers by $\N_0$,
and that of positive integers by $\N_+$.
For a set $S$ and $n,m \in \N_+$, $S^{n \times m}$ is the set of $n$-row,
$m$-column matrices over $S$.
For $n \in \N_+$, $[n] = \{ i \in \N_0 \mid i < n \}$ is the set of the first
$n$ non-negative integers.

Symbols in words are indexed starting with one.
The $i$-th symbol of a word $w$ is denoted by $w_i$.
For an alphabet $\Sigma$ and $n \in \N_0$, $\Sigma^{\le n}$ contains the words
$w \in \Sigma^\ast$ for which $\abs{w} \le n$.
Without restriction, the empty word is not an element of any language that we
consider.

We write $U_n$ (resp., $U_{n \times m}$) for a random variable distributed
uniformly over $\binalph^n$ (resp., $\binalph^{n \times m}$).
For a random variable $X$, $\Supp(X)$ denotes the support of $X$.
We will need the following variant of the Chernoff bound (see, e.g.,
\cite{vadhan12_pseudorandomness_book}):

\begin{theorem}%
  \label{thm_chernoff}%
  Let $X_1,\dots,X_n$ be independently and identically distributed Bernoulli
  variables and $\mu = \Exp[X_i]$.
  Then there is a constant $c > 0$ such that the following holds for every
  $\eps > 0$:
  \[
    \Pr\left[ \abs*{\frac{\sum_{i=1}^n X_i}{n} - \mu} > \eps \right]
    < 2^{-cn\eps^2}.
  \]
\end{theorem}

\subparagraph{Hash functions.}
For $N,M \in \N_+$, a family $H = \{ h\colon [N] \to [M] \}$ of functions is
said to be \emph{pairwise independent} if, for $x_1,x_2 \in [N]$ with $x_1
\neq x_2$ and $h$ chosen uniformly from $H$, the random variables $h(x_1)$ and
$h(x_2)$ are independent and uniformly distributed.
Equivalently, $H$ is pairwise independent if for arbitrary $y_1,y_2 \in [M]$
we have
\[
  \Pr\left[ h(x_1)=y_1 \land h(x_2)=y_2 \right] = \frac{1}{M^2}.
\]
It is a well-known fact (see, e.g., \cite{vadhan12_pseudorandomness_book}) that
there is a family $H$ of pairwise independent functions such that one can
uniformly sample a function $h$ from $H$ with $O(\log N + \log M)$ bits.

\subparagraph{Extractors.}
Let $n \in \N_+$, and let $X$ and $Y$ be random variables taking values in
$\binalph^n$.
Then the \emph{statistical distance} between $X$ and $Y$ is
\[
  \Delta(X,Y) = \frac{1}{2}\sum_w \abs*{\Pr[X=w] - \Pr[Y=w]}.
\]
The \emph{min-entropy} $\Hmin(X)$ of $X$ is defined by
\[
  \Hmin(X) = \min_{w \in \Supp(X)} \log\frac{1}{\Pr[X = w]}.
\]
For $k \le n$, if $\Hmin(X) \ge k$, then $X$ is said to be a \emph{$k$-source}.
For $d,m \in \N_+$ and $\eps > 0$, a \emph{$(k,\eps)$-extractor} is a function
$\Ext\colon\binalph^n \times \binalph^d \to \binalph^m$ such that, for every
$k$-source $X$, we have that $\Delta(\Ext(X,U_d),U_m) \le \eps$.
In this context, the second argument of $\Ext$ is its \emph{seed} and,
correspondingly, $d$ is its \emph{seed length}.


\section{De Bruijn Graphs Fully Characterize SWBPs}
\label{sec_structure}

In this section, we prove:

\restatethmCharSWBP*

\begin{proof}
  We prove first the forward implication.
  Let $S$ be an SWBP of window size $t$.
  The $\alpha_i$ are defined recursively, the basis case being
  $\alpha_0(\lambda)$, which is set to be initial state of $S$.
  Having defined $\alpha_i$ for $i < n$, we set $\alpha_{i+1}$ so that
  \[
    \alpha_{i+1}(wy) = \begin{cases}
      S_i(\alpha_i(w),y), & i < t; \\
      S_i(\alpha_i(xw),y), & i \ge t
    \end{cases}
  \]
  for $x,y \in \binalph$ and $w \in \binalph^{k_i-1}$ (thus automatically
  satisfying the requirements in the statement of the theorem).

  It remains to show that the $\alpha_i$ are well-defined, which we shall prove
  by induction.
  The induction basis $i=0$ is trivial, and the induction step for $i < t$ 
  follows easily from applying the induction hypothesis and $S_i$ being a
  function.
  Hence, suppose $i \ge t$.
  Letting $w \in \binalph^{t-1}$ and $y \in \binalph$, we shall show
  $S_i(\alpha_i(0w),y) = S_i(\alpha_i(1w),y)$.
  Using the induction hypothesis and the properties of the $\alpha_i$, there
  are states $q_0$ and $q_1$ in the $(i-t+1)$-th layer of $S$ so that
  $S_{i-t+1}(q_x,w) = \alpha_i(xw)$ for $x \in \binalph$.
  Thus, by the sliding-window property of $S$ (and, again, by the properties of
  the $\alpha_i$),
  \[
    S_i(\alpha_i(0w),y)
    = S_{i-t+1}(q_0,wy)
    = S_{i-t+1}(q_1,wy)
    = S_i(\alpha_i(1w),y).
  \]

  For the converse implication, let $\alpha_i$ as in the statement be given.
  Then we argue that, for any $i$ and $z \in \binalph^t$ as well as any states
  $q$ and $q'$ that are reachable in the $i$-th layer of $S$, $S_i(q,z) =
  S_i(q',z)$ holds.
  This is simple to see by induction provided that $q$ and $q'$ are in the
  image of $\alpha_i$.
  To see that this is indeed the case, note that, if $q$ is reachable in the
  $i$-th layer by an input $wz \in \binalph^i$ with $\abs{z} = t$, then
  $q = \alpha_i(z)$ (again, due to the sliding-window property) or, in case
  $\abs{z} = 0$ and $i < t$, $q = \alpha_i(w)$.
  In either case, the claim follows.
\end{proof}


\section{Pseudorandom Generators for SWBPs}
\label{sec_general_prg}

In this section, we recall and prove:

\restatethmPRGGenSWBP*

As mentioned in \cref{sec_intro_tech_overview}, in our construction we will
first use $G\base$ to obtain a PRG that simultaneously fools as many unanimity
programs as possible.
Then we use interleaving to convert the resulting construction into a PRG for
SWBPs.
We address the two steps in this order.

\subsection{Simultaneous Fooling}

\begin{definition}%
  \label{def_simul_fool}%
  Let $m,t \in \N_+$ and $\eps > 0$, and let $\mathcal{F}$ be a class of
  functions $f\colon \binalph^t \to \binalph$.
  We say a distribution $X = (X_1,\dots,X_m)$ over $(\binalph^t)^m$
  \emph{$m$-simultaneously $\eps$-fools} $\mathcal{F}$ if the following holds
  for every $f_1,\dots,f_m \in \mathcal{F}$:
  \[
    \abs*{\Pr[\forall i \in [m]: f_i(X_i) = 1]
      - \prod_{i=1}^m \Pr[f_i(U_t) = 1]} \le \eps.
  \]
  Similarly, we say a function $G\colon \binalph^d \to (\binalph^t)^m$
  \emph{$m$-simultaneously $\eps$-fools} $\mathcal{F}$ if $G(U_d)$ fools
  $\mathcal{F}$.
\end{definition}

As already discussed in \cref{sec_intro_tech_overview}, we cover two
possibilities for obtaining a PRG that simultaneously fools unanimity programs
from our base PRG $G\base$.
The first of these bases on the (main idea behind the) INW generator
\cite{impagliazzo94_pseudorandomness_stoc}, which is to stretch the output of
$G\base$ by repeatedly generating a fresh seed for it using an extractor with
adequate parameters.
The second relies on using PRGs that fool combinatorial rectangles to generate
multiple good seeds that are plugged directly into $G\base$.

\subsubsection{Using the INW Generator}

The following is due to the work of \citeauthor{goldreich97_tiny_rsa}
\cite{goldreich97_tiny_rsa} (see also \cite{vadhan12_pseudorandomness_book}):

\begin{theorem}[\cite{goldreich97_tiny_rsa}]%
  \label{thm_gw_extractor}%
  For every $n,k \in \N_+$ and $\eps > 0$, there is a $(k,\eps)$ extractor
  $\Ext\colon\binalph^n \times \binalph^d \to \binalph^n$ with seed length
  $d = O(n - k + \log(1/\eps))$ that is computable in $O(n + d)$ space.
\end{theorem}

We proceed as previously described in \cref{sec_intro_tech_overview}.
Starting from a given generator $G$, we use the extractor to generate a fresh
seed for $G$ using \enquote{just a couple more} random bits and then simply
concatenate the two ouptuts.

\begin{lemma}%
  \label{lem_prg_swbp_ind_step}%
  Let $G\colon \binalph^d \to (\binalph^t)^m$ be a function that
  $m$-simultaneously $\eps$-fools width-$w$, length-$t$ unanimity programs.
  Then there is a function $G'\colon \binalph^{d'} \to (\binalph^t)^{2m}$ with
  $d' = d + O(\log(1/\eps))$ that $2m$-simultaneously $3\eps$-fools width-$w$,
  length-$t$ unanimity programs.
  In addition, if $G$ is explicit, then so is $G'$.
\end{lemma}

\begin{proof}
  Let $\Ext\colon \binalph^d \times \binalph^{d_{\Ext}} \to \binalph^d$ be the
  $(k,\eps)$ extractor of \cref{thm_gw_extractor} where $k = d - \log(1/\eps)$
  and $d_{\Ext} = O(\log(1/\eps))$.
  We set
  \[
    G'(s_G, s_{\Ext}) = G(s_G) G(\Ext(s_G,s_{\Ext}))
  \]
  and denote the $i$-th component in the output of $G'$ by $G'(\cdot)_i$.
  Now let $P_1,\dots,P_{2m}$ be width-$w$, length-$t$ unanimity programs.
  For $j \in \{1,2\}$, let
  \[
    A_j = \{ x_1 \cdots x_m \in \binalph^{mt} \mid
      \text{$x_1,\dots,x_m \in \binalph^t$
        and $\forall i \in [m]: P_{(j-1)m+i}(x_i) = 1$}
    \}
  \]
  and $\mu_j = \mu(A_j)$.
  Observe that $\mu_1\mu_2 = \prod_{i=1}^{2m} \Pr[P_i(U_t) = 1]$.
  If $\mu_1 < 2\eps$, then we immediately have a distance of at most $3\eps$
  between $\mu_1\mu_2 < 2\eps$ and
  \[
    \Pr[\forall i \in [2m]: P_i(G'(U_d)_i) = 1]
    \le \Pr[G(U_d) \in A_1]
    \le \mu_1 + \eps < 3\eps.
  \]
  Hence, suppose that $\mu_1 \ge 2\eps$.
  Observe that, by assumption on $G$, this means
  \[
    \Pr[U_d = x \mid G(U_d) \in A_1]
    \le \frac{\Pr[U_d = x]}{\Pr[G(U_d) \in A_1]}
    \le \frac{2^{-d}}{\mu_1 - \eps}
    \le \frac{2^{-d}}{\eps}
  \]
  for $x \in \binalph^d$.
  In particular, this implies that, if $Z$ is a random variable that is
  distributed according to $\Pr[Z = x] = \Pr[U_d = x \mid G(U_d) \in A_1]$, then
  $\Hmin(Z) \ge k$.
  Thus, by the extraction property,
  \begin{align*}
    \ifconf\MoveEqLeft\fi
    \abs*{\Pr[\forall i \in [2m]: P_i(G'_i(U_d)) = 1] - \mu_1 \mu_2}
    \ifconf\\\fi
    &= \abs*{\Pr[G(U_d) \in A_1 \land G(\Ext(U_d,U_{d_{\Ext}})) \in A_2]
      - \mu_1 \mu_2} \\
    &\le \abs*{\Pr[G(\Ext(U_d,U_{d_{\Ext}})) \in A_2 \mid G(U_d) \in A_1] -
      \mu_2} + \eps \\
    &\le \abs*{\Pr[G(U_d) \in A_2] - \mu_2} + 2\eps \\
    &\le 3\eps.
    \qedhere
  \end{align*}
\end{proof}

Now starting from a PRG that fools (single) width-$w$, length-$t$ unanimity
programs and repeating $r$ times the construction of
\cref{lem_prg_swbp_ind_step}, we obtain:

\begin{lemma}%
  \label{lem_prg_swbp_composed}%
  Let $G\colon \binalph^d \to \binalph^t$ be a function that $\eps$-fools
  width-$w$, length-$t$ unanimity programs.
  For every $r > 0$, there is a function $G'\colon \binalph^{d'} \to
  (\binalph^t)^{2^r}$ with $d' = d + O(r\log(1/\eps))$ that $2^r$-simultaneously
  $3^r \eps$-fools width-$w$, length-$t$ unanimity programs.
  In addition, if $G$ is explicit, then so is $G'$.
\end{lemma}

\subsubsection{Using PRGs for Combinatorial Rectangles}

There has been extensive work on constructing PRGs that fool combinatorial
rectangles.
To obtain the best possible parameters, we will use two different constructions.
The first one is from a classical work by \citeauthor{lu02_improved_comb}:

\begin{theorem}[\cite{lu02_improved_comb}]
  \label{thm_lu_prg}
  For every $n,m \in \N_+$ and $\eps\CR > 0$, there is an explicit PRG
  $G\CR\colon \binalph^{d\CR} \to (\binalph^m)^n$ that $\eps\CR$-fools any
  $(n,m)$-combinatorial rectangle and whose seed length is
  \[
    d\CR = O(m + \log n + (\log(1/\eps\CR))^{3/2}).
  \]
\end{theorem}

The seed length has optimal dependence on $m$, but it is not as good for smaller
values of the error parameter $\eps\CR$.
(The dependence on $n$ is already optimal for our purposes.)
Hence we also make use of a more recent construction by
\citeauthor{gopalan20_concentration_toc}, which gives a better dependence on
$\eps\CR$ at the cost of an additional polylogarithmic factor:

\begin{theorem}[\cite{gopalan20_concentration_toc}]
  \label{thm_gy_prg}
  For every $n,m \in \N_+$ and $\eps\CR > 0$, there is an explicit PRG
  $G\CR\colon \binalph^{d\CR} \to (\binalph^m)^n$ that $\eps\CR$-fools any
  $(n,m)$-combinatorial rectangle and whose seed length is
  \[
    d\CR = O((m + \log\log n + \log(1/\eps\CR)) \log(m + \log(1/\eps\CR))).
  \]
\end{theorem}

\begin{lemma}
  \label{lem_prg_swbp_composed_cr}%
  Let $G\colon \binalph^d \to \binalph^t$ be a function that $\eps$-fools
  width-$w$, length-$t$ unanimity programs.
  For every $r \in \N_+$, there is a function $G'\colon \binalph^{d'} \to
  (\binalph^t)^r$ with either
  \[
    d' = O(d + \log r + (\log(1/\eps\CR))^{3/2})
  \]
  or
  \[
    d' = O((d + \log\log r + \log(1/\eps')) \log(d + \log(1/\eps')))
  \]
  that $r$-simultaneously $\eps'$-fools width-$w$, length-$t$ unanimity
  programs, where $\eps' = O(\eps r)$.
  In addition, if $G$ is explicit, then so is $G'$.
\end{lemma}

\begin{proof}
  We construct $G'$ in two steps.
  First we concatenate $r$ many strings that are generated according to $G$
  using completely random (independent) seeds.
  Then we reduce the seed length from $dr$ to $d'$ by applying either generator
  $G\CR$ from \cref{thm_lu_prg,thm_gy_prg}.
  The correctness follows from our setting of simultaneously fooling $r$ many
  unanimity programs being a special case of fooling combinatorial rectangles.
  Details follow.

  Consider first the intermediate PRG $G\intm\colon (\binalph^d)^r \to
  (\binalph^t)^r$ that is defined by
  \[
    G\intm(s_1,\dots,s_r) = (G(s_1),\dots,G(s_r)).
  \]
  Then we have that $G\intm$ must $r$-simultaneously $(\eps r)$-fool any
  collection of width-$w$, length-$t$ unanimity programs.
  This is because the copies of $G$ use independent seeds, which means that
  \begin{align*}
    \MoveEqLeft
    \abs*{\Pr[\forall i \in [r]: P_i(G(s_i)) = 1]
      - \prod_{i=1}^r \Pr[P_i(U_t) = 1]} \\
    &= \abs*{\prod_{i=1}^r \Pr[P_i(G(s_i)) = 1]
      - \prod_{i=1}^r \Pr[P_i(U_t) = 1]} \\
    &= r \max_{i \in [r]} \abs*{\Pr[P_i(G(s_i)) = 1] - \Pr[P_i(U_t) = 1]} 
      + O(\eps^2),
  \end{align*}
  where the right-hand side is clearly bounded by $O(r \eps)$ (assuming of
  course $\eps < 1$, which is the interesting case).

  We now instantiate the PRG $G\CR$ for $(r,d)$-combinatorial rectangles from
  either one of \cref{thm_lu_prg,thm_gy_prg} with $\eps\CR = \eps r$.
  Composing $G\CR$ with $G\intm$ yields $G' = G\intm \circ G\CR\colon
  \binalph^{d\CR} \to (\binalph^t)^r$.
  To see why this is correct, let $P_1,\dots,P_r$ be any width-$w$, length-$t$
  unanimity programs and consider the combinatorial rectangle $f$ that results
  from composing the $f_i\colon \binalph^d \to \binalph$ with $f_i(s) =
  P_i(G(s))$.
  Then we have that
  \[
    \Pr[\forall i \in [r]: P_i(G'(U_{d\CR})_i) = 1]
    = \Pr[\forall i \in [r]: f_i(G\CR(U_{d\CR})_i) = 1]
  \]
  (where $G'(U_{d\CR})_i$ and $G\CR(U_{d\CR})_i$ denote the $i$-th component in
  $G'(U_{d\CR})$ and $G\CR(U_{d\CR})$, respectively) is $\eps\CR$-close to
  \[
    \Pr[\forall i \in [r]: f_i(U_d) = 1]
    = \Pr[\forall i \in [r]: P_i(G(U_d)) = 1]
  \]
  (where, for each $i$, the $U_d$ are independent copies of the same
  distribution).
\end{proof}

\subsection{Obtaining the Actual PRG}

We will show that, in order to convert the PRGs of
\cref{lem_prg_swbp_composed,lem_prg_swbp_composed_cr} into PRGs for SWBPs, 
it suffices to \emph{interleave} two independent copies of the respective PRG
$G'$.
To this end, for $m,t \in \N_+$ and tuples $x = (x_1,\dots,x_m)$ and $y =
(y_1,\dots,y_m)$ in $(\binalph^t)^m$, let
\[
  x \il y = x_1 y_1 x_2 y_2 \cdots x_m y_m \in \binalph^{2tm}.
\]

\begin{lemma}%
  \label{lem_interleave}%
  Let $m,t \in \N_+$ and $\eps > 0$ be arbitrary.
  Then, for any (independent) random variables $X = (X_1,\dots,X_m)$ and $Y =
  (Y_1,\dots,Y_m)$ taking values in $(\binalph^t)^m$ and such that $X$ (resp.,
  $Y$) $m$-simultaneously $\eps$-fools width-$w$, length-$t$
  unanimity programs, we have that $X \il Y$ (i.e., the random variable which
  assumes values $x \il y$ where $x$ and $y$ are drawn from $X$ and $Y$,
  respectively) $2\eps$-fools width-$w$, length-$2tm$ SWBPs of window size $t$.
\end{lemma}

\begin{proof}
  Let $S$ be a width-$w$, length-$2tm$ SWBP of window size $t$.  We prove that
  $\Pr[S(X \il Y) = 1]$ is $\eps$-close to $\Pr[S(X \il U) = 1]$, where $U$ is a
  random variable that is uniformly distributed on $(\binalph^t)^m$.
  By an analogous argument, we also have that $\Pr[S(X \il U) = 1]$ is
  $\eps$-close to $\Pr[S(U \il U) = 1]$ (where the two occurrences of $U$ denote
  independent copies of the same random variable).
  This then gives the statement of the lemma.

  Fix some $x = (x_1,\dots,x_m) \in \Supp(X)$ and let $S_x^i$ be the width-$w$,
  length-$t$ unanimity program defined as follows:
  \begin{itemize}
    \item $S_x^i$ has the same states as $S$.
    \item The state transition function in the $j$-th layer of $S_x^i$ is the
    same as that in the $(k_i+t+j)$-th layer of $S$, where $k_i = 2(i-1)t$.
    \item The initial state is $S_{k_i}(x_i)$.
    (Note this is well-defined due to $S$ having window size $t$ and
    $\abs{x_i} = t$.)
    \item For $j < t$, the state $s$ is accepting in the $j$-th layer of
    $S_x^i$ if and only if it is accepting in the $(k_i+t+j)$-th layer of $S$.
    A state $s$ in the $t$-th layer of $S_x^i$ is accepting if and only if it is
    also accepting in the $(k_i+2t)$-th layer of $S$ and, in addition,
    $S_{k_i+2t}(x_{i+1}) = 1$.
    (If $i = m$, then this last condition holds vacuously.)
    Furthermore, the initial state of $S_1$ is only accepting if $S_0(x_1) = 1$.
  \end{itemize}
  Note that, by construction, the states that $S$ assumes when reading the $y$
  parts of any input $x \il y$ are the same as the corresponding states of
  $S_x^i$.
  Hence, we observe that $S(x \il y) = 1$ if and only if $S_x^i(y_i) = 1$ for
  every $i$:
  \begin{itemize}
    \item If $S(x \il y) = 0$, then the input $x \il y$ passes through a
    rejecting state in the $r$-th layer of $S$ for some $r \in [n]$.
    \begin{itemize}
      \item If $r$ corresponds to the $y$ part of the input, that is, $r =
      k_i+t+j$ for some $i$ and $j \in [t]$, then (by construction) $y_i$ passes
      through the same state in the $j$-th layer of $S_x^i$.
      \item Conversely, if $r$ corresponds to the $x$ part of the input, then $r
      = k_i+j$ for some $i$ and $j \in [t]$.
      If $i > 1$, then $S_x^{i-1}(y_{i-1}) = 0$ since the state in the last
      layer of $S_x^{i-1}$ rejects; otherwise, $i = 1$ and then $S_x^1(y_1) =
      0$ since the initial state of $S_x^1$ is rejecting.
    \end{itemize}
    In either case, we find some $i$ so that $S_x^i(y_i) = 0$.
    \item If $S_x^i(y_i) = 0$ for some $i$, then there is $j \in [t]$ such that
    $y_i$ passes through a rejecting state in the $j$-th layer of $S_x^i$.
    If the corresponding state in the $(k_i+j)$-th layer of $S$ is also
    rejecting, then $S(x \il y) = 0$, so assume otherwise.
    Then either $j = t$ and then $S_{k+i+2t}(x_{i+1}) = 0$ or $j = 1$ and $i =
    1$, in which case $S_0(x_1) = 0$.
    In either case, $S(x \il y) = 0$ follows.
  \end{itemize}
  By the above and using that $X$ and $Y$ are independent, it follows that
  \begin{align*}
    \MoveEqLeft
    \abs*{\Pr[S(X \il Y) = 1] - \Pr[S(X \il U) = 1]} \\
    &= \abs*{\sum_x \Pr[X=x] \left(\Pr[S(x \il Y) = 1]
        - \Pr[S(x \il U) = 1]\right)} \\
    &\le \sum_x \Pr[X=x] \abs*{\Pr[\forall i \in [m]: S_x^i(Y_i) = 1]
        - \prod_{i=1}^m \Pr[S_x^i(U_t) = 1]} \\
    &\le \eps.
    \qedhere
  \end{align*}
\end{proof}

We now round up the ideas above to prove \cref{thm_prg_gen_swbp}.

\begin{proof}[Proof of \cref{thm_prg_gen_swbp}]
  We instantiate two copies of the generator $G\base$ with independent seeds,
  stretch each to at least $n/2$ bits using either of
  \cref{lem_prg_swbp_composed,lem_prg_swbp_composed_cr}, and then combine them
  by interleaving blocks of $t$ bits each.
  Details follow.

  For simplicity, we assume $n$ is a multiple of $2t$.
  For the first item in the theorem, we plug in $G\base$ with $r = \log(n/2t)$
  in \cref{lem_prg_swbp_composed} to obtain a generator $G'\colon \binalph^{d'}
  \to (\binalph^t)^{n/2t}$ with
  \[
    d' = d\base + O(\log(n/t)\log(1/\eps\base))
  \]
  that $(n/2t)$-simultaneously $\eps'$-fools width-$w$, length-$t$ unanimity
  programs, where $\eps' = \eps\base \cdot (n/2t)^{\log 3}$.
  For the second one, again we plug in $G\base$ in
  \cref{lem_prg_swbp_composed_cr} but now setting $r = n / 2t$.
  This yields a generator $G'\colon \binalph^{d'} \to (\binalph^t)^{n/2t}$ with
  either
  \[
    d' = O(d\base + \log(n/t) + (\log(1/\eps\base))^{3/2})
  \]
  or
  \[
    d' = O((d\base + \log\log(n/t) + \log(1/\eps\base))
      \log(d\base + \log(1/\eps\base)))
  \]
  that $(n/2t)$-simultaneously $\eps'$-fools width-$w$, length-$t$ unanimity
  programs, where $\eps' = O(\eps\base \cdot n / t)$.

  For the actual construction of the generator $\Gfunc$ from the claim, we shall
  use two copies of $G'$ (with independent seeds), which for convenience we
  denote $G_1$ and $G_2$.
  In turn, for $i \in \{1,2\}$, the components in the output of $G_i(\cdot)$ are
  denoted $G_i(\cdot)_1,\dots,G_i(\cdot)_{n/2}$.
  Then $G$ simply interleaves the outputs of $G_1$ and $G_2$; that is,
  \[
    G(s_1,s_2) = G_1(s_1)_1G_2(s_2)_1G_1(s_1)_2G_2(s_2)_2 \cdots
      G_1(s_1)_{n/2t}G_2(s_2)_{n/2t}
  \]
  for $s_1,s_2 \in \binalph^{d'}$.
  Hence, the seed length of $G$ is $2d'$, and naturally $G$ is explicit.
  The correctness of $G$ follows in both cases directly from
  \cref{lem_interleave} combined with either \cref{lem_prg_swbp_composed}
  or \cref{lem_prg_swbp_composed_cr}.
\end{proof}


\section{Hitting Set Generators for SWBPs}
\label{sec_hsg}

In this section, we prove:

\restatethmHSGSWBP*

The proof of \cref{thm_hsg_swbp} is just a straightforward adaptation of that of
the second item of \cref{thm_prg_gen_paca}.
To avoid repetition, we state only the facts that are particularly relevant to
the HSG case and omit parts of the argument that are identical to what was
already covered in \cref{sec_general_prg}.

The first notion we need to relax is simultaneous fooling.

\begin{definition}%
  \label{def_simul_hit}%
  Let $m,t \in \N_+$ and $\eps > 0$, and let $\mathcal{F}$ be a class of
  functions $f\colon \binalph^t \to \binalph$.
  We say a distribution $X = (X_1,\dots,X_m)$ over $(\binalph^t)^m$
  \emph{$m$-simultaneously $\eps$-hits} $\mathcal{F}$ if the following holds
  for every $f_1,\dots,f_m \in \mathcal{F}$:
  \[
    \Pr[\forall i \in [m]: f_i(U_t) = 1] \ge \eps
    \implies
    \exists (x_1,\dots,x_m) \in \Supp(X): \forall i \in [m]: f_i(x_i) = 1
  \]
  (where, for each $i$, $U_t$ is an independent copy of the same random variable).
  As before, we say a function $G\colon \binalph^d \to (\binalph^t)^m$
  \emph{$m$-simultaneously $\eps$-hits} $\mathcal{F}$ if $G(U_d)$ fools
  $\mathcal{F}$.
\end{definition}

Instead of the PRG from \cref{thm_gy_prg}, we will apply the HSG due to
\citeauthor{linial97_efficient_comb}.
As mentioned in \cref{sec_intro_tech_overview}, its seed length is
asymptotically optimal.

\begin{theorem}[\cite{linial97_efficient_comb}]
  \label{thm_llsz_hsg}
  For every $n,m \in \N_+$ and $\eps\CR > 0$, there is an explicit $\eps\CR$-HSG
  $G\CR\colon \binalph^{d\CR} \to (\binalph^m)^n$ for $(n,m)$-combinatorial
  rectangles whose seed length is
  \[
    d\CR = O(m + \log\log n + \log(1/\eps\CR)).
  \]
\end{theorem}

Our equivalent of \cref{lem_prg_swbp_composed_cr} is now the following:

\begin{lemma}
  \label{lem_hsg_swbp_composed_cr}%
  Let $G\colon \binalph^d \to \binalph^t$ be a function that $\eps$-fools
  width-$w$, length-$t$ unanimity programs.
  For every $r \in \N_+$, there is a function $G'\colon \binalph^{d'} \to
  (\binalph^t)^r$ with
  \[
    d' = O(d + \log\log r + \log(1/\eps'))
  \]
  that $r$-simultaneously $\eps'$-hits width-$w$, length-$t$ unanimity programs,
  where $\eps' = O(\eps r)$.
  In addition, if $G$ is explicit, then so is $G'$.
\end{lemma}

\begin{proof}
  Let $G'$ and $G\intm$ be as in the proof of \cref{lem_prg_swbp_composed_cr}
  but where we use the HSG $G\CR$ from \cref{thm_llsz_hsg} instead of the PRG
  from \cref{thm_gy_prg}.
  To be more precise, we require $G\CR$ to be an $(\eps r)$-HSG for
  $(d,r)$-combinatorial rectangles.
  Since the respective part of the construction is identical, we still have that
  $G\intm$ must $r$-simultaneously $(\eps r)$-fool any collection of width-$w$,
  length-$t$ unanimity programs.
  Hence, we need only show the correctness of $G'$ when replacing the seeds of
  $G\intm$ with those generated by $G\CR$.

  Again, let $P_1,\dots,P_r$ be any width-$w$, length-$t$ unanimity programs and
  define $f_i$ as in the proof of \cref{lem_prg_swbp_composed_cr}, and suppose
  that $\Pr[\forall i \in [r]: P_i(U_t) = 1] \ge 2 \eps r$ (where the $U_t$ are
  independent copies of the same distribution).
  By the property of $G\intm$ stated above, this means that
  \[
    \Pr[\forall i \in [r]: f_i(U_d) = 1] \ge \eps r.
  \]
  The statement then follows directly from $G\CR$ being an $(\eps r)$-HSG.
\end{proof}

The analogue of \cref{lem_interleave} is the following:

\begin{lemma}%
  \label{lem_interleave}%
  Let $m,t \in \N_+$ and $\eps > 0$ be arbitrary.
  Then, for any (independent) random variables $X = (X_1,\dots,X_m)$ and $Y =
  (Y_1,\dots,Y_m)$ taking values in $(\binalph^t)^m$ and such that $X$ (resp.,
  $Y$) $m$-simultaneously $\eps$-hits width-$w$, length-$t$ unanimity programs,
  we have that $X \il Y$ (i.e., the random variable which assumes values $x \il
  y$ where $x$ and $y$ are drawn from $X$ and $Y$, respectively) $2\eps$-hits
  width-$w$, length-$2tm$ SWBPs of window size $t$.
\end{lemma}

\begin{proof}
  The proof is in some sense similar to that of \cref{lem_interleave}.
  Let $S$ and $U$ be as in said proof and suppose that $\Pr[S(U \il U) = 1] \ge
  2\eps$.
  We use the same reduction (to the $S_x^i$), but the correctness requires a
  different argument since what we must now show is that $S(x \il y) = 1$ for
  some $x \in \Supp(X)$ and $y \in \Supp(Y)$ (instead of just arguing that the
  distribution $S(X \il Y)$ is close to $S(U \il U)$).

  The key fact to observe is that the event of all $S_x^i$ accepting (for any
  choice of $x$) is subsumed under $S$ itself accepting.
  Concretely, this means that, for any fixed $x$ such that $\Pr[S(x \il U) = 1]
  \ge \eps$, we have that, if we sample $x$ according to $U$, then
  \[
    \Pr_x[\forall i \in [m]: S_x^i(U_t) = 1]
    \ge \Pr[S(x \il U) = 1]
    \ge \eps,
  \]
  implying there is $y = y(x) = (y_1,\dots,y_m) \in \Supp(Y)$ so that
  $S_x^i(y_i) = 1$ for every $i$ and thus $S(x \il y) = 1$.

  Let us call an $x$ that satisfies this condition (i.e., $\Pr[S(x \il U) = 1]
  \ge \eps$) \emph{good}.
  By Markov's inequality, 
  \[
    \Pr_{x \gets U}[\Pr[S(x \il U) = 0] \ge 1 - \eps]
    \le \frac{1 - 2\eps}{1 - \eps} < 1 - \eps,
  \]
  which means that a uniformly chosen $x$ is good with probability strictly
  greater than $\eps$.
  In turn, as argued above, every such $x$ is such that there is $y \in
  \Supp(Y)$ with $S(x \il y) = 1$.
  It follows that $\Pr[S(U \il Y) = 1] > \eps$.

  This brings us halfway to our goal (i.e., of showing there are $x \in
  \Supp(X)$ and $y \in \Supp(Y)$ so that $S(x \il y) = 1$).
  For the next step, we again use Markov's inequality to obtain
  \[
    \Pr_{y \gets Y}[\Pr[S(U \il y) = 0] \ge 1 - \eps] < 1.
  \]
  Hence, the probability that $\Pr[S(U \il y) = 1] \ge \eps$ holds is non-zero
  if $y$ is chosen according to $Y$.
  Using an analogous argument to the above, it follows there is $x = x(y) \in
  \Supp(X)$ so that $S(x \il y) = 1$.
\end{proof}

The final details of the proof of \cref{thm_hsg_swbp} are as in that of
\cref{thm_prg_gen_paca}.
This concludes the proof of \cref{thm_hsg_swbp}.


\section{Application to Sublinear-Time Probabilistic Cellular Automata}
\label{sec_paca}

For the results in this section, we assume the reader is familiar with the
theory of cellular automata.
(See, e.g., \cite{delorme99_cellular_book} for a standard reference.)

\subsection{Probabilistic Cellular Automata}
\label{sec_def_paca}

As mentioned in the introduction, the PACA model was defined in
\cite{modanese23_sublinear-time_stacs}.
We repeat here the definitions verbatim for the reader's convenience.

We consider only bounded one-dimensional cellular automata.

\begin{definition}%
  \label{def_ca}%
  A \emph{cellular automaton} is a triple $C = (Q,\$,\delta)$ where
  $Q$ is the finite set of \emph{states},
  $\$ \notin Q$ is the \emph{boundary symbol},
  and $\delta\colon Q_\$ \times Q \times Q_\$ \to Q$ is the \emph{local
    transition function}, where $Q_\$ = Q \cup \{ \$ \}$.
  The elements in the domain of $\delta$ are the possible \emph{local
    configurations} of the cells of $C$.
  For a fixed width $n \in \N_+$, the \emph{global configurations} of $C$ are
  the elements of $Q^n$.
  The cells $1$ and $n$ are the \emph{border cells} of $C$.
  The \emph{global transition function} $\Delta\colon Q^n \to Q^n$ is obtained
  by simultaneous application of $\delta$ everywhere; that is, if $s \in Q^n$
  is the current global configuration of $C$, then
  \[
    \Delta(s) = \delta(\$,s_1,s_2)
      \, \delta(s_1,s_2,s_3)
      \, \cdots
      \, \delta(s_{n-1},s_n,\$).
  \]
  For $t \in \N_0$, $\Delta^t$ denotes the $t$-th iterate of $\Delta$.
  For an initial configuration $s \in Q^n$, the sequence $s = \Delta^0(s),
  \Delta(s), \Delta^2(s), \dots$ is the \emph{trace} of $C$ (for $s$).
  Writing the trace of $C$ line for line yields its \emph{space-time diagram}.
  Finally, for a cell $i \in [n]$ and $r \in \N_0$, the cells in $[i-r,i+r] \cap
  [n]$ form the \emph{$r$-neighborhood} of $i$.
\end{definition}

\begin{definition}%
  \label{def_aca}%
  A \emph{deterministic ACA} (DACA) is a cellular automaton $C$ with an
  \emph{input alphabet} $\Sigma \subseteq Q$ as well as a subset $A \subseteq Q$
  of \emph{accepting states}.
  We say $C$ \emph{accepts} an input $x \in \Sigma^+$ if there is $t \in \N_0$
  such that $\Delta^t(x) \in A^n$, and we denote the set of all such $x$ by
  $L(C)$.
  In addition, $C$ is said to have \emph{time complexity} (bounded by) $T\colon
  \N_+ \to \N_0$ if, for every $x \in L(C) \cap \Sigma^n$, there is $t <
  T(\abs{x})$ such that $\Delta^t(x) \in A^n$.
\end{definition}

\begin{definition}%
  \label{def_PACA}%
  Let $\Sigma$ be an alphabet and $Q$ a finite set of states with $\Sigma
  \subseteq Q$.
  A \emph{probabilistic ACA} (PACA) $C$ is a cellular automaton with two local
  transition functions $\delta_0, \delta_1\colon Q^3 \to Q$.
  At each step of $C$, each cell tosses a fair coin $c \in \binalph$ and updates
  its state according to $\delta_c$; that is, if the current configuration of
  $C$ is $s \in Q^n$ and the result of the cells' coin tosses is $r = r_1 \cdots
  r_n \in \binalph^n$ (where $r_i$ is the coin toss of the $i$-th cell),
  then the next configuration of $C$ is
  \[
    \Delta_r(s) = \delta_{r_1}(\$,s_1,s_2) \, \delta_{r_2}(s_1,s_2,s_3) \,
      \cdots \, \delta_{r_n}(s_{n-1},s_n,\$).
  \]
  Seeing this process as a Markov chain $M$ over $Q^n$, we recast the global
  transition function $\Delta = \Delta_{U_n}$ as a family of random variables
  $(\Delta(s))_{s \in Q^n}$ parameterized by the current configuration $s$ of
  $C$, where $\Delta(s)$ is sampled by starting in state $s$ and performing a
  single transition on $M$ (having drawn the cells' coin tosses according to
  $U_n$).
  Similarly, for $t \in \N_0$, $\Delta^t(s)$ is sampled by starting in $s$ and
  performing $t$ transitions on $M$.

  A \emph{computation} of $C$ for an input $x \in \Sigma^n$ is a path in $M$
  starting at $x$.
  The computation is \emph{accepting} if the path visits $A^n$ at least once.
  In addition, in order to be able to quantify the probability of a PACA
  accepting an input, we additionally require for every PACA $C$ that there is a
  function $T\colon \N_+ \to \N_0$ such that, for any input $x \in \Sigma^n$,
  every accepting computation for $x$ visits $A^n$ for the first time in
  strictly less than $T(n)$ steps; that is, if there is $t \in \N_0$ with
  $\Delta^t(x) \in A^n$, then $\Delta^{t_1}(x) \in A^n$ for some $t_1 < T(n)$.
  (Hence, every accepting computation for $x$ has an initial segment with
  endpoint in $A^n$ and whose length is strictly less than $T(n)$.)
  If this is the case for any such $T$, then we say $C$ has \emph{time
    complexity} (bounded by) $T$.

  With this restriction in place, we may now equivalently replace the coin
  tosses of $C$ with a matrix $R \in \binalph^{T(n) \times n}$ of bits with
  rows $R_0,\dots,R_{T(n)-1}$ and such that $R_j(i)$ corresponds to the coin
  toss of the $i$-th cell in step $j$.
  (If $C$ accepts in step $t$, then the coin tosses in rows $t,\dots,T(n)-1$
  are ignored.)
  We refer to $R$ as a \emph{random input} to $C$.
  Blurring the distinction between the two perspectives (i.e., online and
  offline randomness), we write $C(x,R) = 1$ if $C$ accepts $x$ when its coin
  tosses are set according to $R$, or $C(x,R) = 0$ otherwise.
\end{definition}

\begin{definition}%
  \label{def_pPACA}%
  Let $L \subseteq \Sigma^\ast$ and $p \in [0,1)$.
  A \emph{one-sided $p$-error PACA for $L$} is a PACA $C$ with time complexity
  $T$ such that, for every $x \in \Sigma^n$, the following holds:
  \begin{align*}
    x \in L &\iff \Pr[C(x, U_{T(n) \times n}) = 1] \ge 1-p
    \qquad \text{and}
    \ifconf \\ \else \qquad \fi
    x \notin L &\iff \Pr[C(x, U_{T(n) \times n}) = 1] = 0.
  \end{align*}
  If $p = 1/2$, then we simply say $C$ is a \emph{one-sided error PACA}.
  Similarly, for $p < 1/2$, a \emph{two-sided $p$-error PACA for $L$} is a PACA
  $C$ such that, for every $x \in \Sigma^\ast$, the following holds:
  \begin{align*}
    x \in L &\iff \Pr[C(x, U_{T(n) \times n}) = 1] \ge 1-p
    \qquad \text{and}
    \ifconf \\ \else \qquad \fi
    x \notin L &\iff \Pr[C(x, U_{T(n) \times n}) = 1] \le p.
  \end{align*}
  If $p = 1/3$, then we simply say $C$ is a \emph{two-sided error PACA}.
  In both cases, we write $L(C) = L$ and say $C$ \emph{accepts} $L$.
\end{definition}

\subsection{Simulating a PACA with a Low-Space Sliding-Window Algorithm}
\label{sec_simulation}

We now show how a PACA can be simulated by a randomized sliding-window algorithm
with low space.
This can be achieved with little difficulty simply by adapting the streaming
algorithm from \cite{modanese21_lower_csr}.
(The algorithm there is actually geared toward a different variant of cellular
automata, but the same strategy can be applied to our setting.)
For the sake of self-containedness, we provide the adaptation in full.

\begin{proposition}%
  \label{prop_rand_str_alg}%
  Let $C$ be a (one- or two-sided error) PACA with state set $Q$ and $T \in
  \N_0$.
  Then there is a $O(T \log\abs{Q})$-space randomized non-uniform sliding-window
  algorithm $S = S_T$ of window size $O(T^2)$ such that
  \[
    \Pr[S(x) = 1] = \Pr[\text{$C$ accepts $x$ in time step $T$}].
  \]
\end{proposition}

We note the non-uniformity of $S$ is required only to set $T$.
Every other aspect of $S$ is realized in an uniform manner.

The basic idea involved is that, in order to emulate the behavior of $C$ on an
input $x$, it suffices to move a sliding window over its time-space diagram that
is $T$ cells long and $O(1)$ cells wide (see \cref{fig_sim_str_alg}), feeding
random bits to the cells as needed.
Every time a new symbol from $x$ is read, the window is moved one cell to the
right; if it positioned beyond the borders of $C$, then the respective part of
the window is filled with the border symbol $\$$.
If $S$ notes that any one cell in step $T$ is not accepting, then it immediately
halts and rejects; otherwise it reads the entire input and eventually accepts.

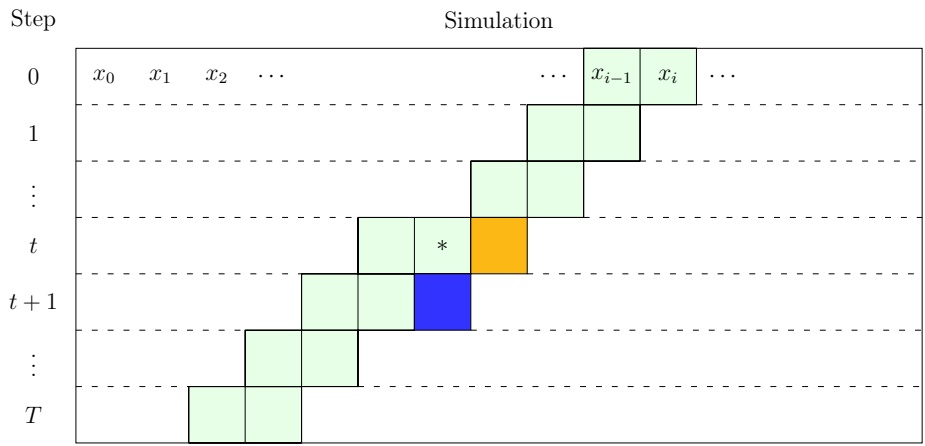
\begin{figure}
  \centering
  \begin{tikzpicture}
  \draw (0,0) rectangle (15,7);

  \foreach \i in {0,...,3}
  \draw[fill=lightgreen] (2+\i,0+\i) grid (4+\i,1+\i) rectangle (2+\i,0+\i);

  \foreach \i in {0,...,2}
  \draw[fill=lightgreen] (7+\i,4+\i) grid (9+\i,5+\i) rectangle (7+\i,4+\i);

  \draw[fill=yellow!70!red] (7,3) rectangle (8,4);
  \draw[fill=blue!80!] (6,2) rectangle (7,3);

  \node at (0.5,6.5) {$x_0$};
  \node at (1.5,6.5) {$x_1$};
  \node at (2.5,6.5) {$x_2$};
  \node at (3.5,6.5) {$\cdots$};
  \node at (8.5,6.5) {$\cdots$};
  \node at (9.5,6.5) {$x_{i-1}$};
  \node at (10.5,6.5) {$x_i$};
  \node at (11.5,6.5) {$\cdots$};

  \node at (7.5,7.5) {Simulation};


  \node at (-0.75,7.5) {Step};
  \node at (-0.75,6.5) {$0$};
  \node at (-0.75,5.5) {$1$};
  \node at (-0.75,4.5) {$\vdots$};
  \node at (-0.75,3.5) {$t$};
  \node at (-0.75,2.5) {$t+1$};
  \node at (-0.75,1.5) {$\vdots$};
  \node at (-0.75,0.5) {$T$};

  \node at (6.5,3.5) {$\ast$};

  \foreach \i in {1,...,4}
    \draw[dash pattern=on 3pt off 5pt] (0,\i) -- (\i+1,\i);
  \foreach \i in {5,6}
    \draw[dash pattern=on 3pt off 5pt] (0,\i) -- (\i+2,\i);
  \foreach \i in {1}
    \draw[dash pattern=on 3pt off 5pt] (\i+4,\i) -- (15,\i);
  \foreach \i in {2,...,6}
    \draw[dash pattern=on 3pt off 5pt] (\i+5,\i) -- (15,\i);
\end{tikzpicture}
  \caption{Simulation of $C$ by the sliding-window algorithm $S$.
    In the picture, $S$ has last read the input symbol $x_i$ and is now
    determining the state of the cell marked with an asterisk in time step
    $t+1$ (depicted in blue).
    The light green cells are the ones maintained by $S$ in its
    \textsf{stateLeft}
    and
    \textsf{stateCenter}
    arrays, while the orange cell is the one corresponding to
    \textsf{stateRight}.
  }
  \label{fig_sim_str_alg}
\end{figure}
\begin{algorithm}[t]
  \For{$t \gets 0, \dots, T-1$}{
    $\StateLeft[t] \gets \$$\;
    $\StateCenter[t] \gets \$$\;
  }
  $i \gets 0$\;
  \lnl{line_outer_while}%
  \While{$i < n+T$}{
    \uIf{$i < n$}{
      $\StateRight \gets x_i$\;
    }
    \Else{
      $\StateRight \gets \$$\;
    }
    $t \gets 0$\;
    \lnl{line_inner_while}%
    \While{$t < T$}{
      \lnl{line_if_acc}%
      sample $b \in \binalph$ from randomness source\;
      $\NewState \gets
        \delta_b(\StateLeft[t], \StateCenter[t], \StateRight)$\;
      $\StateLeft[t] \gets \StateCenter[t]$\;
      $\StateCenter[t] \gets \StateRight$\;
      $\StateRight \gets \NewState$\;
      $t \gets t+1$\;
    }
    \If{$\StateRight \neq \$$ and $\StateRight$ is not accepting}{
      \Reject;
    }
    $i \gets i+1$\;
  }
  \Accept\;
  \caption{Randomized sliding-window algorithm $S$}
  \label{alg_rand_stream}
\end{algorithm}

\begin{proof}
  We prove that \cref{alg_rand_stream}, hereafter referred to as $S$, satisfies
  the properties in the claim.
  For simplicity of presentation, in this proof the input $x = x_0 \cdots
  x_{n-1}$ as well as the cells of the PACA $C$ are indexed starting with zero.

  We first address the correctness of $S$.
  Fix a random input $r \in \binalph^{(n+T)T}$ to $S$.
  It shall be convenient to recast $r$ as a matrix $R \in \binalph^{T\times
    (n+T)}$ where $R(i,j) = r(i + jT)$; that is, the $j$-th column of $R$
  corresponds to the randomness used by $S$ in the $j$-th run of the outer
  \KwSty{while} loop, and the $i$-th row of said column equals the value of $b$
  in the $i$-th run of the inner \KwSty{while} loop.
  In addition, let $R' \in \binalph^{T \times n}$ be the matrix with $R'(i,j) =
  R(i,i+j+1)$.
  As we shall see, $R'$ corresponds exactly to the random input to $C$ in its
  simulation by $S$ and the bits of $R$ that do not have a corresponding entry
  in $R'$ do not affect the outcome of $S$.

  By $D = D_{R'}\colon \{0,\dots,T\} \times \Z \to Q_\$$ we denote the
  time-space diagram of $C$ when using coin tosses from $R'$ where $\$$ is used
  to fill states \enquote{beyond the borders} of $C$; that is, for $i \in [n]$,
  $D(t,i)$ equals the state of the $i$-th cell in the $t$-th step of $C$ on
  input $x$ when using the coin tosses given by $R'$, or $D(t,i) = \$$
  otherwise.
  We shall show the following invariants are satisfied by the outer loop
  (\cref{line_outer_while}):
  \begin{enumerate}
    \item[$I_1$:] For every $t \in \{0,\dots,T-1\}$, $\StateCenter[t] =
    D(t,i-t-1)$ and $\StateLeft[t] = D(t,i-t-2)$.
    \item[$I_2$:] $S$ has not rejected if and only if, for every $i' < i$,
    $D(T,i'-T)$ is either accepting or equal to $\$$.
  \end{enumerate}
  In particular, $I_2$ directly implies the correctness of $S$ (if we also have
  that $S$ depends only on the random bits of $R$ that have a corresponding
  entry in $R'$, which is indeed the case).

  We prove $I_1$ and $I_2$ by showing the inner loop (\cref{line_inner_while})
  satisfies an invariant $I_3$ of its own, namely that $\StateRight = D(t,i-t)$.
  Concretely, if we have that $I_1$ and $I_2$ hold prior to the $i$-th execution
  of the outer loop and $I_3$ holds at the end of the inner loop, then $I_1$ and
  $I_2$ are conserved as follows:
  The invariant $I_1$ holds since the instructions executed ensure that
  $\StateLeft[t] = D(t,i-t-1)$ and $\StateCenter[t] = D(t,i-t)$ hold for every
  $t$ at the end of the inner loop and $i$ is incremented at the end.
  Similarly, since $t = T$ holds after the inner loop is done, we have then
  $\StateRight = D(t,i-T)$ after the loop, implying $I_2$.

  To show $I_3$ is an invariant, suppose $I_1$ holds for some $i$.
  Clearly, $\StateRight = D(0,i)$ holds prior to its first execution of the loop
  since then $\StateRight = x(i) = D(0,i)$ if $i < n$, or $\StateRight = \$ =
  D(0,i)$ otherwise.
  Subsequently, in the $t$-th execution of the loop, if $t+1 \le i \le n+t$,
  then $b$ is set to $R(t,i) = R'(t,i-t-1)$ and \StateRight to
  \begin{align*}
    \delta_b(\StateLeft[t], \StateCenter[t], \StateRight)
    &= \delta_b(D(t,i-t-2), D(t,i-t-1), D(t,i-t)) \\
    &= D(t+1,i-t-1).
  \end{align*}
  If $i < t+1$, then $i-t-1 < 0$, and so $\StateLeft[t] = \StateCenter[t] = \$$,
  which means $\StateRight$ is set to $\$$ regardless of the value of $R(t,i)$;
  the same is the case for $i > n+t$ since then $i-t-1 > n-1$.
  Hence, the operation of $S$ depends only on $R(t,i)$ if $t+1 \le i \le n+t$,
  which is precisely the case when $R'(t,i-t-1) = R(t,i)$.
  Since $t$ is incremented at the end of the loop, it follows that $I_3$ is an
  invariant, as desired.

  The space complexity of $S$ evident since it is dominated by the arrays
  \StateLeft and \StateRight, which both contain $T$ many elements of $Q$.
  Finally, regarding the sliding window size of $S$, we can argue based on the
  invariants above and the properties of the time-space diagram $D$.
  Clearly, an entry $R'(t,i)$ only affects the states $D(t+j,i+k)$ for $j \in
  \{1,\dots,T-t\}$ and $k \in \Z$ with $\abs{k} < j$ (since, for every $t$ and
  every $i$, $R'(t,i)$ only affects $D(t+1,i)$ and $D(t,i)$ only $D(t+1,i-1)$,
  $D(t+1,i)$, and $D(t+1,i+1)$).
  Hence, after having read $R'(t,i)$, if $S$ reads another $2T^2$ random bits,
  then its state will be independent of the entry $R'(t,i)$.
  This means that $S$ has a window size of $O(T^2)$, as desired.
\end{proof}

\subsection{Derandomizing Sublinear-Time PACAs with Small Space}

In this final section, we recall and prove our low-space derandomization results
for the languages accepted by sublinear-time PACA (i.e.,
\cref{thm_prg_gen_paca,thm_prg_gen_paca_twosided}).

\subsubsection{One-Sided Error PACAs}

\restatethmPRGGenPACA*

The proof is more or less straightforward:
Given a $T$-time one-sided $\eps$-error PACA $C$ as in the theorem's statement,
we use the construction from \cref{prop_rand_str_alg} together with the HSG from
\cref{cor_thm_hsg_swbp}.
Since the PACA may accept in any time step $t < T$, we need to simulate the
algorithm $S_t$ from \cref{prop_rand_str_alg} for every such step $t$; however,
by an averaging argument, there is at least one $t$ so that $S_t$ accepts with
probability at least $\eps / T$.
Hence, it suffices to use an $(\eps / T)$-HSG.

\begin{proof}
  We describe an algorithm $A$ for $L(C)$ with the desired space complexity.
  Let $G$ be the PRG from \cref{cor_thm_prg_gen_swbp} that $\eps_G$-fools SWBPs
  of width $2^{O(T)}$, length $m = (n+T)T$, and window size $O(T^2)$, where
  $\eps_G = \eps/T$.
  In addition, for $t < T$, let $S_t$ be the sliding-window algorithm of
  \cref{prop_rand_str_alg}.
  Our algorithm $A$ operates as follows:
  \begin{enumerate}
    \item For every $t < T$, enumerate over every possible seed to $G$ and
    simulate $S_t$ on its output.
    \item If $S_t$ accepts, accept immediately.
    \item If none of the $S_t$ have accepted, reject.
  \end{enumerate}
  Observe the space complexity of $A$ is dominated by the seed length of $G$.
  By our setting of parameters, this is
  \[
    d = O\left(
      \frac{(T + \log(1/\eps)) \log T}{\max\{1, \log T - \log\log(n/\eps)\}}
      + \log(1/\eps)
      + \log n
    \right)
  \]
  (and certainly greater than $O(T)$, which is the space required to simulate
  each run of $S_t$).

  For the correctness of $A$, fix some input $x$ to $C$.
  If $x \notin L(C)$, then the probability that $C$ accepts $x$ is zero, which
  means $\Pr[S_t(U_m) = 1] = 0$ for every $t$, thus implying $A$ never accepts.
  Otherwise, if $x \in L(C)$, by an averaging argument, there is $t$ such that
  $\Pr[S_t(U_m) = 1] \ge \eps/T = \eps_G$.
  Since $G$ is an $\eps_G$-HSG, there is a seed $s \in \binalph^d$ so that
  $S_t(G(s))$ accepts, and thus $A$ also accepts $x$.
\end{proof}

\subsubsection{Two-Sided Error PACAs}
\label{sec_paca_darand_twosided}

For two-sided error PACAs, ideally we would like to simply \enquote{replace} the
HSG with a PRG in the proof of \cref{thm_prg_gen_paca}.
After all, the PRG would essentially allow us to estimate the probability that
the PACA accepts in every single time step $t$.
Unfortunately, this fails even if we do have a means of \emph{exactly}
determining these probabilities.

Let us explain why by means of an example.
For concreteness, let the input alphabet be $\Sigma = \binalph$.
Consider the PACAs $C_1$ and $C_2$ that operate as follows:
Every cell except for the first one is unconditionally in an accepting state.
During the first 8 steps, the leftmost cell collects random bits to form a
random string $r \in \binalph^8$.
Following this, in the next steps, the cell behaves as follows:
\begin{itemize}
  \item In $C_1$, the cell turns accepting if and only if the first two bits of
  $r$ are equal to zero (i.e., $r_1 = r_2 = 0$).
  The cell remains accepting for exactly four steps, then changes into a
  non-accepting state and maintains it.
  \item In $C_2$, in each of the subsequent four steps $j \in [4]$, the cell
  turns accepting in step $j$ if and only if $r_{2j-1} = r_{2j} = 0$.
  After this, the cell assumes a non-accepting state and maintains it.
\end{itemize}
It is easy to see that both $C_1$ and $C_2$ are two-sided error PACAs that
accept completely different languages:
The probability that $C_1$ accepts any input $w \in \binalph^\ast$ is $1/4$,
which means $L(C_1) = \varnothing$; conversely, $C_2$ accepts any input with
probability $1 - (3/4)^4 > 2/3$, and so $L(C_2) = \binalph^\ast$.
Nevertheless, the probability that the PACA accepts at any \emph{fixed} time
step $t$ (considered on its own) is \emph{the same} in either $C_i$.

This means that a new strategy is called for.
One possible solution would be to adapt the algorithm $S$ from
\cref{prop_rand_str_alg} so that it checks multiple steps of the PACA for
acceptance (e.g., maintain a Boolean variable $\mathtt{accept}_i$ for every step
$i$ initially set to true and set it to false if any non-accepting state in step
$i$ is seen).
It is unclear how to do so without forgoing the sliding-window property,
however.
Nevertheless, this strategy \emph{can} be applied (without giving up on the
sliding-window property) if instead of a single time step $t$ we have $S$ verify
that the PACA is accepting in every step that is in a \emph{subset} $t \subseteq
[T]$ (that is provided to $S$ non-uniformly).
Indeed, this is a straightforward adaptation and it is not hard to see that the
resulting algorithm $S_t$ is such that
\[
  \Pr[S_t(x) = 1] = \Pr[\forall i \in t: \text{on input $x$, $C$ is accepting in
  time step $i$} ]
\]
where the second probability is conditioned on the coin tosses of $C$.
(Under this aspect, the original algorithm works in the special case where
$\abs{t} = 1$.)
Then we can resort to the inclusion-exclusion principle to rephrase the
probability of $C$ accepting as a sum over the probabilities of the $S_t$
accepting (for every possible non-empty choice of $t$).
Using a PRG with adequate parameters, we can then estimate every term in this
sum and decide whether $C$ accepts or not based on whether this estimate exceeds
or is below $1/2$.

\restatethmPRGGenPACATwosided*

\begin{proof}
  We construct an algorithm $A$ for $L(C)$ with the desired space complexity.
  Let $G$ be the PRG from item 2 of \cref{cor_thm_prg_gen_swbp} that
  $\eps_G$-fools SWBPs of width $2^{O(T)}$, length $m = (n+T)T$, and window size
  $O(T^2)$, where $\eps_G = \eps/2^T$.
  (These are the same parameters as in the proof of \cref{thm_prg_gen_paca}
  except for $\eps_G$, where the additional factor is now exponentially small in
  $T$.)
  For $\varnothing \neq t \subseteq [T]$, let $S_t$ be the generalization of the
  algorithm of \cref{prop_rand_str_alg} as previously described.
  The algorithm $A$ proceeds as follows:
  \begin{enumerate}
    \item For every $\varnothing \neq t \subseteq [T]$, enumerate over every
    possible seed to $G$ and simulate $S_t$ on its output to obtain an estimate
    $\eta_t = \Pr[S_t(G(U_d)) = 1]$ that is $\eps_G$-close to $\Pr[S_t(U_m) =
    1]$.
    \item Compute
    \[
      \eta = \sum_{\stackrel{t \subseteq [T]}{\abs{t} = 1}} \eta_t
      - \sum_{\stackrel{t \subseteq [T]}{\abs{t} = 2}} \eta_t
      + \cdots 
      + (-1)^{T+1} \eta_{[T]}
    \]
    and accept if $\eta > 1/2$; otherwise reject.
  \end{enumerate}
  By our setting of parameters, $G$ has seed length
  \[
    d = O((d' + \log n) \log(d' + \log n)),
  \]
  where
  \[
    d' = O((T + \log(1/\eps)) \log T).
  \]
  As in the proof of \cref{thm_prg_gen_paca}, the space complexity of $A$ is
  dominated by $d$.

  For the correctness, let $Z_t(R)$ denote the event that $C$ accepts $x$ in
  every one of the steps in $\varnothing \neq t \subseteq [T]$ when using coin
  tosses from $R \in \binalph^{T \times n}$.
  Note that, by the inclusion-exclusion principle, we have
  \begin{align*}
    \ifconf\MoveEqLeft\fi
    \Pr[C(x,U_{T \times n}) = 1]
    \ifconf\\\fi
    &= \Pr[\exists t \in [T]: Z_{\{t\}}(U_{T \times n})] \\
    &= \sum_{\stackrel{t \subseteq [T]}{\abs{t} = 1}} \Pr[Z_t(U_{T \times n})]
    - \sum_{\stackrel{t \subseteq [T]}{\abs{t} = 2}} \Pr[Z_t(U_{T \times n})]
    + \cdots 
    + (-1)^{T+1} \Pr[Z_{[T]}(U_{T \times n})].
  \end{align*}
  Since there are strictly less than $2^T$ terms on the right-hand side in total
  and replacing $U_{T \times n}$ with $G(U_d)$ causes an additive deviation of
  at most $\eps_G$ in each term, it follows that
  \[
    \abs*{\eta - \Pr[C(x,U_{T \times n} = 1)]}
    \le \eps_G 2^T
    = \eps. \qedhere
  \]
\end{proof}


\ifconf\else
\paragraph{Acknowledgments.}
\myack
\fi

\printbibliography


\end{document}